\documentclass[journal]{new-aiaa}

\usepackage[utf8]{inputenc}
\usepackage{textcomp}
\usepackage[linesnumbered, ruled, vlined]{algorithm2e}

\SetCommentSty{mycommfont}

\newcommand{\bs}[1]{\boldsymbol{#1}}

\usepackage{graphicx}
\usepackage{amsmath}
\usepackage[version=4]{mhchem}
\usepackage{siunitx}
\usepackage{booktabs}
\usepackage{graphicx}
\usepackage{multirow}
\usepackage{longtable,tabularx}
\usepackage[english]{babel}

\usepackage{amsthm}
\newtheorem{theorem}{Theorem}[section]

\usepackage{subcaption}

\usepackage[capitalize]{cleveref}
\usepackage{cprotect}

\usepackage{tablefootnote}
\usepackage{placeins}

\title{Two-Impulse Trajectory Design in Two-Body Systems \\ With Riemannian Geometry}

\author{Samuel G. Gessow\footnote{Ph.D. Student, Department of Mechanical and Aerospace Engineering, VECTR Lab, sgessow@ucla.edu, AIAA Student Member.}, 
James Tseng\footnote{Graduate Student, Department of  Mechanical and Aerospace Engineering, VECTR Lab, jtseng22@ucla.edu, AIAA Student Member.}, 
Eden Zafran\footnote{Graduate Student, Department of  Mechanical and Aerospace Engineering, VECTR Lab, edenzafran@ucla.edu, AIAA Student Member.}, 
and Brett T. Lopez \footnote{Assistant Professor, Department of Mechanical and Aerospace Engineering, VECTR Lab, btlopez@ucla.edu, AIAA Member.}}
\affil{University of California Los Angeles, Los Angeles, CA, 90025}

\begin{document}

\maketitle

\begin{abstract}
This work presents a new method for generating impulsive trajectories in restricted two-body systems by leveraging Riemannian geometry.
The proposed method transforms the standard trajectory optimization problem into a purely geometric one that involves computing a set of geodesics for a suitable Riemannian metric.
This transformation is achieved by defining a metric, specifically the Jacobi metric, that embeds the dynamics directly into the metric, so any geodesic of the metric is also a dynamically feasible trajectory.
The method finds the fuel-optimal transfer trajectory by sampling candidate energy ($\Delta V$) changes for different points on the current and desired orbit, and efficiently computing and evaluating each candidate geodesic, which are equivalent to candidate orbit transfer trajectories via the Jacobi metric.
The method bypasses the known issues of optimization-based methods, e.g., sensitivity to the initial guess, and can be applied to more complex two-body systems.
The approach is demonstrated on the minimum-$\Delta V$ two-impulse phase-free orbit transfer problem, first on a Keplerian system and second on a system with a modeled $J_2$ perturbation. 
The proposed method is shown to meet or exceed the state-of-the-art methods in the minimum-$\Delta V$ problem in the Keplerian system.
The generality and versatility of the approach is demonstrated by seamlessly including the $J_2$ perturbation, a case that many existing methods cannot handle. 
Numerical simulations and performance comparisons showcase the effectiveness of the approach.
\end{abstract}

\section{Introduction}
\lettrine{F}{inding} fuel-optimal orbit transfer trajectories is a fundamental problem in space mission planning and is often expressed as finding a maneuver that minimizes the overall change in velocity ($\Delta V$) required by the spacecraft.
In classical trajectory design \cite{shirazi2018spacecraft,walsh2020survey}, it is assumed a spacecraft can generate instantaneous $\Delta V$ through an impulsive maneuver.
This reduces the orbit transfer problem to determining when, where, and by how much $\Delta V$ needs to change.
While modern trajectory design methods are now focused on continuous thrust maneuvers \cite{shirazi2018spacecraft}, impulsive trajectories are still used to provide a reasonable initial guess that is further refined by a numerical solver \cite{saloglu2025impulsive}.
Despite some of the first works on impulsive trajectory design dating back more than a century \cite{hohmann1925erreichbarkeit,mccue1963optimum}, and with some recent developments \cite{koblick2019robust}, algorithms for designing impulsive trajectories still often rely on numerical methods sensitive to the quality of the initial guess, leading to converge to local minima or complete diverge of the numerical solver \cite{betts1998survey}.
This work develops a new impulsive trajectory design framework that reformulates the trajectory design problem as a purely geometric curve finding problem.
In doing so, we bypass the sensitives observed in other methods while also obtaining a method that generalizes to more challenging/complex problems, such as non-spherical bodies.

Recently, there has been a growing interest and success in applying geometric concepts to trajectory planning and control \cite{kaptui2022fuel,adu2025bring}. 
In particular, new methods are being developed using geodesics, which are defined as the shortest path between two points on a manifold in Riemannian geometry \cite{jaquier2022riemannian}, to represent dynamically feasible trajectories. 
If this conversion from trajectories to geodesics can be made, then several efficient solvers for computing geodesics can be employed \cite{leung2017nonlinear,kaptui2022fuel,adu2025bring} that bypass the issues with numerical optimization in the time domain.
Fundamentally, geodesics are the shortest curve connecting any two points on a manifold given a distance function, i.e., metric.
To convert a trajectory planning problem into a geodesic computation problem, a metric must be found that naturally embeds the dynamics so that any geodesic is also a dynamically feasible trajectory.
Certain classes of systems have known metrics whose geodesics describe state trajectories, the most relevant for this work being Lagrangian systems with an accompanying Jacobi metric.
This work will utilize this connection between geometry and dynamically feasible trajectories to design impulsive spacecraft trajectories that minimize a performance index. 

The main contribution of this paper is a new geometric framework for computing minimum-$\Delta V$ two-impulse phase-free orbital transfers that generalizes easily to more complex gravity models while avoiding local minima and the need for an accurate initial guess.
Specifically, the method leverages geodesics of a Riemannian metric to compute dynamically feasible trajectories.
Central to the approach is the geometric formulation of the phase-free transfer problem, which enables efficient computation and circumvents the convergence issues common in traditional methods.
The key insight of this work is that spacecraft operating in a restricted two-body system possesses the so-called Jacobi metric whose geodesics are also dynamically feasible trajectories.
Combining this insight with a sampled-based planning algorithm allows the efficient computation of fuel optimal transfer trajectories even for more complex restricted two-body problems, such as those with a modeled $J_2$ perturbation.
Specifically, while prior work has addressed $J_2$-perturbed transfers through iterative corrections \cite{ma2017global,zhou2013optimal}, this framework directly incorporates gravitational perturbations such as $J_2$ into the metric itself, without relying on precomputed unperturbed solutions.
The method is demonstrated on both a classical Keplerian system, where it meets or exceeds the performance of existing techniques, and on a $J_2$ perturbed system, showcasing its flexibility.
More broadly, this approach lays the groundwork for extending geometric methods to increasingly complex mission environments.

\section{Differential Geometry Preliminaries}
We first give a brief review of essential concepts from differential geometry as it is fundamental to our approach.
Let $\mathbb{S}_+^n$ denote the set of $n\times n$ positive definite matrices. 
Furthermore, let $\mathcal{M}$ be a smooth manifold (which will be $\mathbb{R}^n$ for this work) equipped with a Riemannian metric $\mathcal G: \mathcal{M} \rightarrow \mathbb{S}_+^n$ that defines an inner product $\langle \cdot, \cdot \rangle_{\bs{x}}$ on the tangent space $T_{\bs{x}} \mathcal M$ at every point $\bs{x} \in \mathcal{M}$.
The Riemannian metric $\mathcal G(\bs{x})$ defines local geometric notions such as angles, length, and orthogonality at every point $\bs{x} \in \mathcal{M}$.
Given two points $\bs{p}_0,\,\bs{p}_f \in \mathcal{M}$, let $\bs{c}:[0,1] \rightarrow \mathcal{M}$ be a regular (i.e., $\partial \bs{c} / \partial s =\bs{c}_s \neq 0 ~  \forall s \in [0,1]$) parametrized differentiable curve such that $\bs{c}(0)=\bs{p}_0$ and $\bs{c}(1)=\bs{p}_f$.
The length $\mathcal{L}$ and Riemannian energy $\mathcal{E}$ of curve $\bs{c}$ are given by
\begin{equation}
    \mathcal{L}=\int_0^1 \sqrt{\bs{c}_s^\top \mathcal{G}(\bs{c})\bs{c}_s}ds, \ \ \mathcal{E}=\int_0^1{\bs{c}_s^\top \mathcal{G}(\bs{c})\bs{c}_s}ds.
\end{equation}
Let $\Xi(\bs{p}_0,\bs{p}_f)$ denote the family of regular curves with $\bs{c}(0)=\bs{p}_0$ and $\bs{c}(1)=\bs{p}_f$.
The Riemannian distance between points $\bs{p}_0$ and $\bs{p}_f$ is given by
\begin{equation}
\label{eq:geo_min}
    d(\bs{p}_0,\bs{p}_f) = \inf_{c(s)\in\Xi}\mathcal{L}(\bs{c}).
\end{equation}
By the Hopf-Rinow theorem, under suitable conditions a minimizing curve known as a \textit{minimum geodesic} $\bs{\gamma}:[0,1] \rightarrow \mathcal{M}$ is guaranteed to exist with the unique property $\mathcal{E}(\bs{\gamma})=\mathcal{L}(\bs{\gamma})^2 \leq \mathcal{L}(\bs{c})^2 \leq \mathcal{E}(\bs{c})$.
In other words, a minimum geodesic is the shortest curve between two points on a manifold.
A minimum geodesic can also be thought of as the extension of the concept of a straight line, which is the minimum geodesic in Euclidean space.
In addition to the minimum geodesic, any curve satisfying the geodesic equation
\begin{equation}
    \nabla_{\bs{\gamma}_s}\bs{\gamma}_s  = 0,
    \label{eq:geo}
\end{equation}
is a \textit{geodesic}, where $\nabla$ is the covariant derivative. The covariant derivative generalizes the directional derivative to spaces whose basis vectors change direction or length throughout the space \cite{do2016differential}.
\Cref{eq:geo} will be expressed in a more useful form (i.e., in coordinates) in \cref{sec:computing_geo}.
Note that any minimum geodesic is a geodesic, but a geodesic may not be minimal in the sense of \cref{eq:geo_min} because there can be several extremal curves that satisfy \cref{eq:geo}.
For the remainder of the paper, we will omit the ``minimal'' qualifier when discussing geodesics.

\section{Spacecraft Trajectories and Geometry}
Spacecraft trajectory design can be formulated in the time or spatial domain.
The time domain is often chosen because the equations of motion are almost always derived with time as the independent variable.
The primary numerical procedures for generating impulsive trajectories to transfer a spacecraft from position $\bs{p}_0$ to position $\bs{p}_f$ are the so-called indirect (shooting) methods \cite{betts1998survey}.
Indirect methods (see \cref{fig:shooting}) iteratively guess the instantaneous change in velocity ($\Delta V$) required to reach $\bs{p}_f$ from $\bs{p}_0$.
The $\Delta V$ guess is refined by computing the error in final position by integrating the dynamics forward in time given the current guess, then iterating to minimize the error.
However, indirect methods are known to be hypersensitive and require a good initial guess to perform reasonably \cite{betts1998survey}. 

An alternative formulation is to eliminate the time dependence and design the \emph{shape} of the trajectory connecting points $\bs{p}_0$ and $\bs{p}_f$, posing the trajectory design problem as a geometry problem of determining a minimum length curve, i.e., a minimal geodesic, that connects $\bs{p}_0$ and $\bs{p}_f$. 
The benefits of formulating the trajectory design problem in this way will be clarified later in the manuscript, but it is important to note here that one cannot completely bypass the dynamics of the system. 
Rather than express the dynamics as an explicit constraint that must be satisfied, in certain cases the dynamics can actually be encoded into the geometry of the problem. 
One way of encoding the dynamics is to select a metric $\mathcal{G}(\bs{q})$, with $\bs{q}$ being the configuration coordinates so that the geodesics of the metric are valid trajectories of the system.
Equipped with such a metric, the shooting problem of finding an initial condition can be converted into the geometric problem of finding a geodesic.
There are many methods to compute a geodesic, one method used in this paper relies on systematically deforming a curve into the geodesic as shown in \cref{fig:deform} and discussed more in \cref{sec:computing_geo}.
This method is more stable than shooting methods and will converge without a good initial guess.

For certain classes of systems, there exists a systematic way of constructing a Riemannian metric to capture the dynamics of the system called the Jacobi metric \cite{pin1975curvature,arnol2013mathematical}.
One class is Lagrangian systems of the form $L(\bs{q},\bs{\dot q})=\frac{1}{2} \bs{\dot q}^\top \mathcal Q(\bs{q})\bs{\dot q} -V(\bs{q})$, where $\mathcal{Q}$ is the mass tensor (and a Riemannian metric) and $V$ is the potential function.
Systems of this form are sometimes referred to as ``natural Lagrangian systems.''
Using the Mauritius principle \cite{pin1975curvature,yourgrau2012variational}, the Jacobi metric 
$\mathcal{G}(\bs{q})=(2E-V(\bs{q}))\mathcal{Q}(\bs{q})$ is derived such that geodesics of this metric describe motion of the system, where the derivation is shown in \ref{sec:appendix_derivation_jacobi}.
In other words, the geodesics of the Jacobi metric are paths that the system will follow if energy of the system is conserved.
The Jacobi metric for the standard Keplerian system is shown in \cref{sec:ex_kep}, and the metric for the system with a $J_2$ perturbation is shown in \cref{sec:J2}.
In addition, the other states of the trajectory, such as velocity, acceleration, and elapsed time, can all be computed from a geodesic once one is found.
For $s \in [0,1]$, the time along the trajectory is given as
\begin{equation}
t(s) = \int_{0}^{s} \frac{1}{2\bigl(E - V\bigl(\bs{q}(\sigma)\bigr)\bigr)} d\sigma.
\end{equation}
The rest of the states can be computed using $\bs{q}(t(s))$ and taking derivatives.

\begin{figure}[t!]
    \begin{subfigure}{.5\textwidth}
        \centering
        \includegraphics[width = \textwidth, trim={30 0 50 13}, clip,]{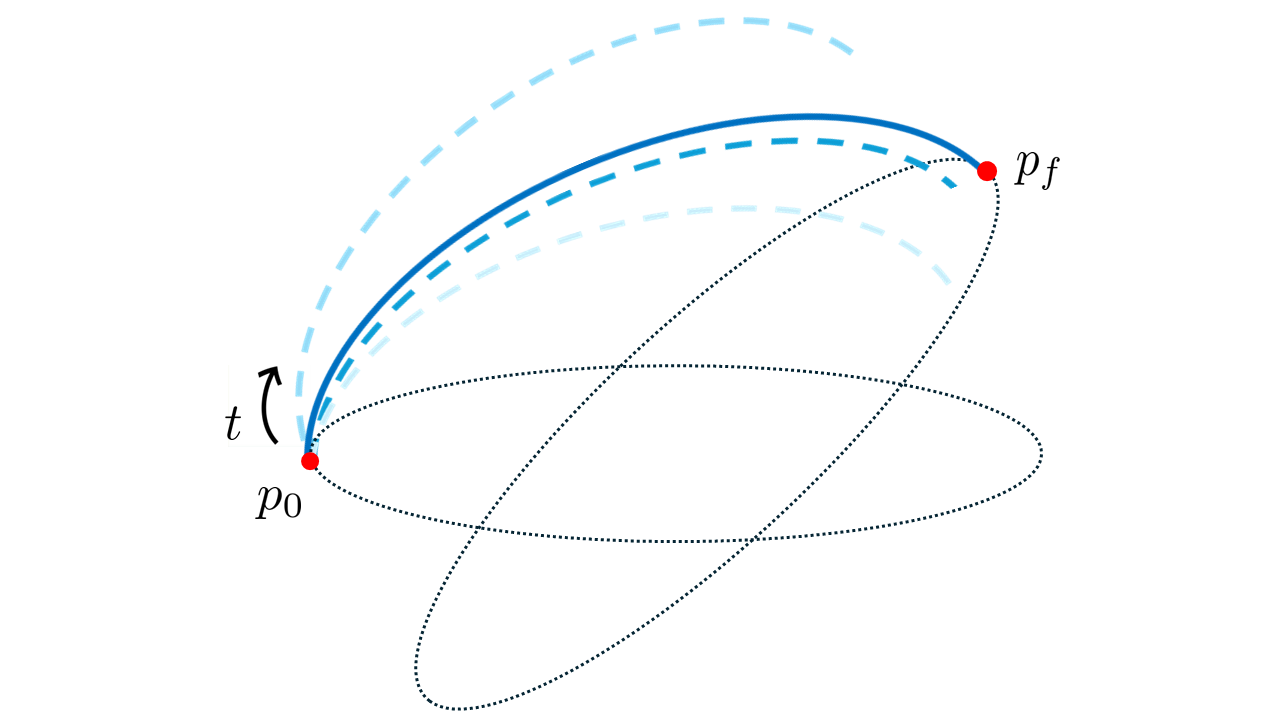}
             \caption{Shooting method.}
             \label{fig:shooting}
     \end{subfigure}
     \begin{subfigure}{.5\textwidth}
         \centering \includegraphics[width = \textwidth, trim={35 27 50 10}, clip]{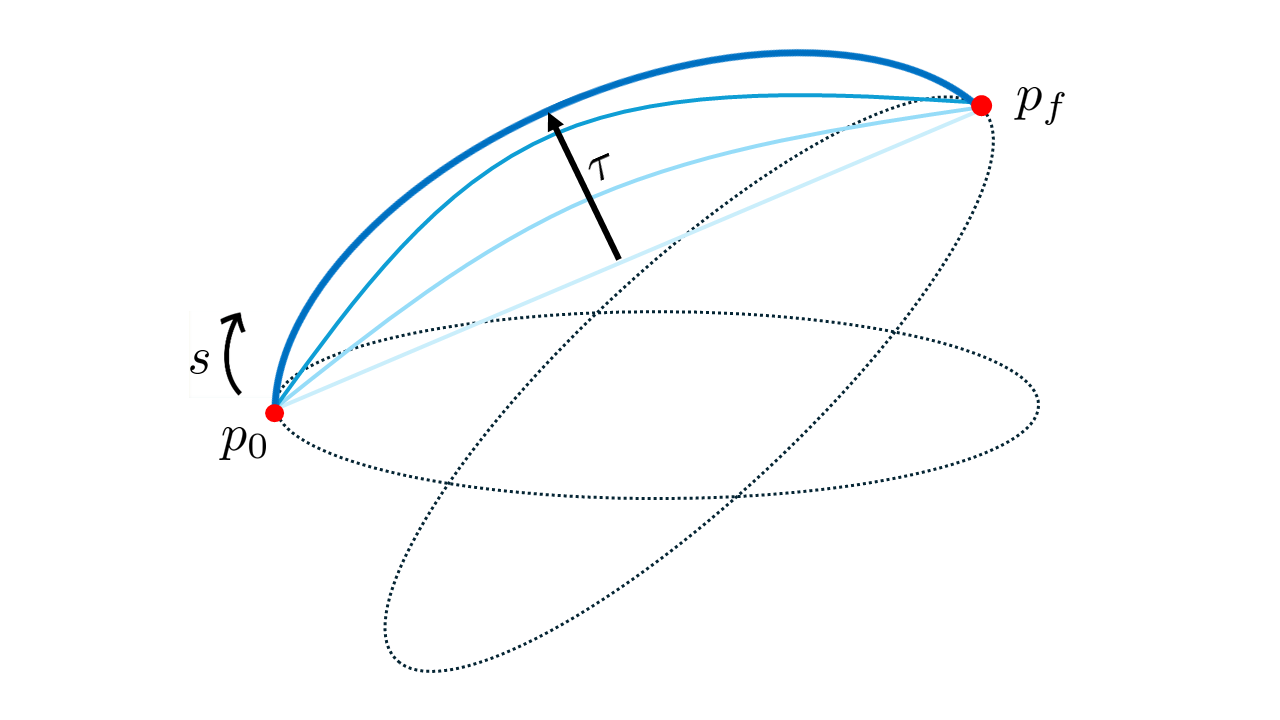}
         \caption{Curve deformation.}
         \label{fig:deform}
     \end{subfigure}
     \caption{Illustration of shooting method and curve deformation evolution.}
\end{figure}

\section{Algorithm: Sampling-Based Trajectory Design with Geodesics}
\label{sec:methods}
In this section, we present a new trajectory planning framework using geodesics for orbit transfer optimization given an arbitrary cost function. 
Although the geodesic of the Jacobi metric is the shortest path for the metric and a feasible path for the system, it is not necessarily the optimal trajectory for an arbitrary cost function. 
This can be seen when the cost function is minimizing $\Delta V$ for the orbital transfer problem, as the vector initial and final velocities of the geodesic are non-linear functions of the Riemannian energy, $E$, and are not minimized when $E$ is the lowest. 
To optimize for this cost function, we must therefore search over multiple geodesics of varying $E$ values and compute the resulting $\Delta V$.

Designing transfer trajectories using Riemannian geometry has several key steps, all of which are summarized in \cref{algorithm}.
Given a suitable Riemannian metric (in this paper, the Jacobi metric) and performance index, e.g., minimum $\Delta V$, one must compute a sample set that can be used to compute geodesics.
For the Jacobi metric, three quantities are needed to compute a geodesic, namely the initial position $\bs{p}_0$, final position $\bs{p}_f$, and energy level $E$ of the path. 
For a two-impulse orbital transfer problem, sampling over the starting point on the initial orbit, final point on the target orbit, and energy yields a collection of transfer orbits of varying energies between two points specified on the initial and target orbit. 
Illustrated in \cref{fig:orbit_sample}, this shows a many-to-one sampling in position for three energies, which is expanded to many-to-many sampling in the actual algorithm we used.
After computing the geodesics for all sampled parameters, the results are searched over to find the geodesic that best minimizes the cost function. 
Because of the nature of sampling, the true minimum is likely not included in the sampling, but the optimal with respect to the sampled set is guaranteed to be found.
The computed solution can be refined using a narrower but denser sample set, or can be used to warm start an optimization method.
We expand on each step in \cref{algorithm} in the following subsections.

\begin{figure}[t!]
    \centering
    \includegraphics[height=.3\textheight, trim={130 60 40 70}, clip,]{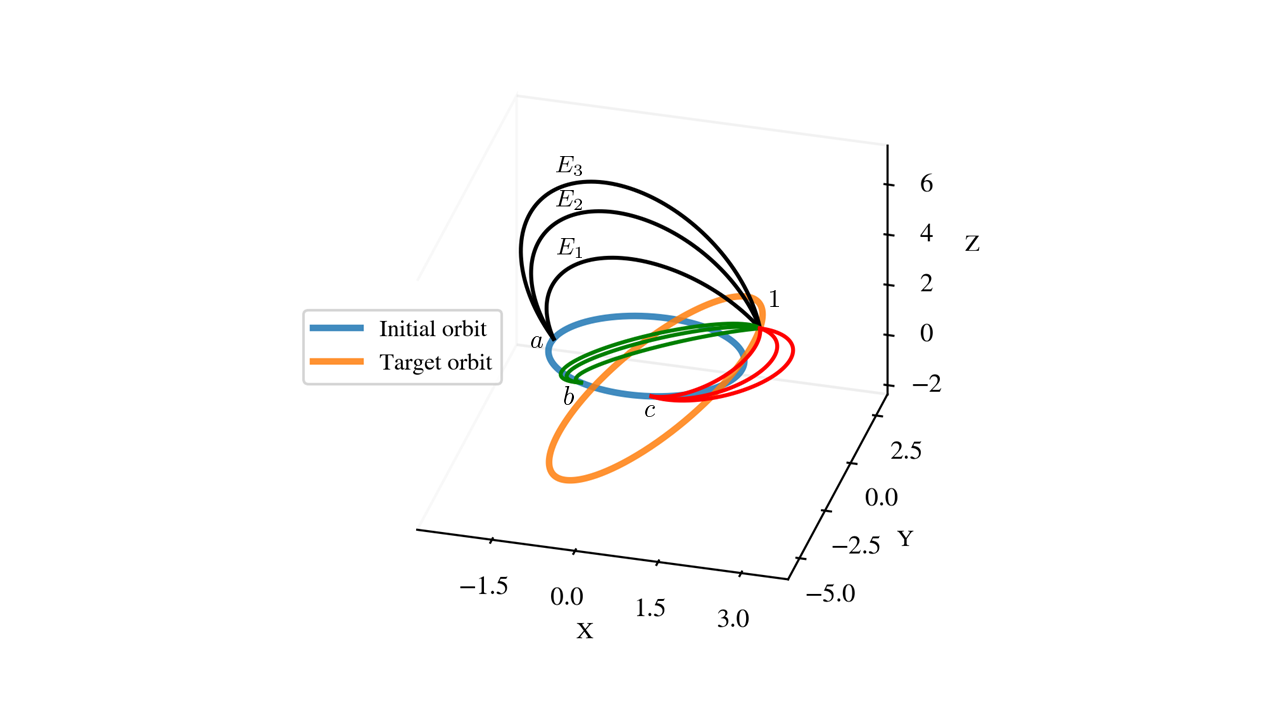}
    \caption{Sampling over position of initial ($a$, $b$, $c$) and target ($1$) orbits \& transfer orbit energies $E_1$, $E_2$, $E_3$.}
    \label{fig:orbit_sample}
\end{figure}

\begin{algorithm}[t]

\caption{Geodesic Trajectory Planning} \label{algorithm}
\small \SetAlgoLined

\KwIn{$\mathcal{G} \gets$ Jacobi metric, \ $\mathcal{J} \gets$ cost function}
\KwOut{$\mathbf{X}_{\text{opt}}$ optimal trajectory}

\BlankLine

\footnotesize {\tcp{sample parameters and compute geodesics}}
$\mathcal{A}, \mathcal{B}, \mathcal{C}  \gets$ getSampleParameters() \\
\For{$\bs{p}_0 \in \mathcal{A}$} {
    \For{$\bs{p}_f \in \mathcal{B}$} {
        \For{ $E \in \mathcal{C}$} {
            $\mathbf{X}_{\bs{p}_0,\bs{p}_f,E} \gets$ computeGeodesic($\mathcal{G},\bs{p}_0,\bs{p}_f,E$) \\
        }
    }
}

\footnotesize {\tcp{search over and evaluate geodesics using cost function}}
$\mathbf{X}_{\text{best}} \gets$ evaluateGeodesics($\mathcal{J},\mathbf{X}$)

\footnotesize {\tcp{refine solution using optimization}}
$\mathbf{X}_{\text{opt}} \gets$ refineSolution($\mathcal{G}, \mathcal{J},\mathbf{X}_{\text{best}}$)

\end{algorithm}

\subsection{Sampling}
Geodesics for a given Jacobi metric do not directly minimize a prescribed cost function, so it is necessary to sample three parameters in the orbit transfer problem: the position on the initial orbit $\bs{p}_0$, the position on the target orbit $\bs{p}_f$, and the energy value of the geodesic $E$ to find the best geodesic for the objective.
In this paper, we uniformly sampled the position coordinate, defined in a Cartesian frame $xyz$. 
Establishing appropriate bounds on the energy parameter $E$ is essential for efficient sampling.
Low energy levels may cause disconnected regions in the configuration space \cite{koon2000dynamical}, corresponding to non-existent geodesics connecting certain initial and final points. 
Moreover, excessively large energy levels can also waste computation resources as the resulting geodesics require large $\Delta V$, making them unlikely to be useful for the objective. 
In certain cases, such as the one explored in \cref{sec:ex_kep}, bounds on the minimum energy can be derived.
In general, loose bounds can often be obtained and used with minimal wasted computation accrued when no geodesic exists.

One method to improve performance is a coarse-to-fine sampling procedure, where the geodesic trajectory planning algorithm is run with coarse sampling over positions $xyz$, then rerun with finer sampling in a region around the optimum found under the coarse sampling regime.
This method is successful if the basin around the minimum is large, where the coarse sampling will pick a point in the correct basin.
This naturally leads to a more general optimization scheme for refining the sampled solution, explored more in \cref{sec:refine}.
Effective sampling is key to performance, as the overall method will only perform well if the sampled set includes the (near-)optimal element.

\subsection{Computing Geodesics}
\label{sec:computing_geo}
Many methods have been proposed to compute geodesics for a Riemannian metric \cite{leung2017nonlinear}. 
A recent approach, particularly in the context of motion planning, is to solve the geodesic as the solution to a geometric heat flow partial differential equation (PDE) \cite{chou2001curve, liu2019affine}. 
In this method, the geodesic is posed as the steady-state solution of a PDE formulated in terms of the Christoffel symbols of the metric.
The geometric heat flow PDE for finding geodesics from \cite{jost2008riemannian} and later used in \cite{liu2019affine} is given by
\begin{equation}
\frac{\partial \bs{c}(s, \tau)}{\partial \tau} = \nabla_{\bs{c}_s(s, \tau)} \bs{c}_s(s, \tau),
\label{eq:pde}
\end{equation}
where $\bs{c}(s, \tau)$ is a family of curves parameterized by $s$ and an additional dummy variable $\tau$, ${\bs{c}_s}(s,\tau) = \frac{\partial \bs{c}}{\partial s}$ is the derivative with respect to $s$, and $\nabla$ is the covariant derivative.
This equation can also be seen as a PDE extension of \cref{eq:geo}. 
This PDE can be solved using a traditional PDE solver or by collocation methods.
For the example in \cref{sec:J2}, a pseudospectral approach with Chebyshev polynomials was implemented to solve \cref{eq:pde}.

In coordinates, \cref{eq:pde} expands to
\begin{equation}
\frac{\partial c^i}{\partial \tau} = \frac{\partial^2 c^i}{\partial s^2} + \Gamma^i_{jk}(c) \frac{\partial c^j}{\partial s} \frac{\partial c^k}{\partial s},
\label{eq:pde_christoffel}
\end{equation}
where we employ Einstein notation\footnote{Einstein notation follows the rule that whenever an index appears once up and once down it is implicitly summed. For example $a_i b^i = \sum_i a_i b^i$.} and $c^i$ is the $i^\text{th}$ component of $\bs{c}$.
$\Gamma^i_{jk}$ are the Christoffel symbols of the metric, defined by the following equation
\begin{equation}
\Gamma^i_{jk} = \frac{1}{2} g^{il} \left( \frac{\partial g_{lj}}{\partial c^k} + \frac{\partial g_{lk}}{\partial c^j} - \frac{\partial g_{jk}}{\partial c^l} \right),
\label{eq:christoffel_def}
\end{equation}
with the Riemannian metric  $g_{ij}$ and its inverse $g^{ij}$, where $g_{ij}$ is the $i,j$ component of $\mathcal{G}$.  
All terms in these equations depend only on the metric $\mathcal{G}$ and its derivatives, and are easily computed once $\mathcal{G}$ is defined.
This method of computing geodesics has theoretical guarantees of existence and can still converge even if the initial guess is far from the optimal. 
Specifically, on a compact Riemannian manifold, every homotopy class of closed curves contains a geodesic, and this result can extend to geodesically complete manifolds \cite{jost2008riemannian}.
Even on non-compact and non-geodesically complete manifolds, convergence failures of the heat flow-based method are typically predictable (e.g. the initial curve goes through a singularity) and can often be mitigated by perturbing the initial curve.
In some situations, such as the restricted two-body problem with Keplerian orbits, geodesics can instead be computed more efficiently by exploiting the known structure of the geometry of the space, which is explored further in \cref{sec:ex_kep}.

\subsection{Geodesic Evaluation \& Selection}
Once the geodesics have been computed for all the sampled parameters, they can be evaluated with respect to the performance index, with the best selected for execution. 
The cost function in this framework is arbitrary and can even be changed online, since cost evaluation is typically computationally inexpensive.
For example, if a solution with excessive time of flight is found in the minimum $\Delta V$ transfer trajectory determination problem, the cost function can be adjusted to one that includes a penalty or bound on time of flight.
For the orbit transfer problems considered in this work, i.e., two-impulse trajectories in a two-body system, geodesic evaluation and selection is essentially a fast enumeration method.
Future work will expand this baseline algorithm to conduct a search over geodesics for multi-impulse maneuvers in multi-body systems.

\subsection{Refining the Solution}
\label{sec:refine}
As mentioned previously, the computed trajectory is only as good as the sampling density around the optimum.
To address this issue, the best solution or $n$-best solutions can be further refined via optimization or denser sampling over a narrower parameter range.
This further refinement is beneficial as it gets the benefits of sampling, namely, robustness to local minima, without the drawback of missing the true optimal. 
This work does not focus on improving optimization-based algorithms used for refining the solution, but we did demonstrate the efficacy of combining some off the shelf optimization methods into our proposed framework.

\section{Simulation Results}
\label{sec:results}
The trajectory planning algorithm outlined in \cref{sec:methods} is evaluated on two test cases. 
The first test case of the minimum $\Delta V$ phase-free orbital transfer in the Keplerian system demonstrates that our proposed method matches or exceeds existing methods on a well studied problem, thereby validating the approach.
The second test case of the minimum $\Delta V$ phase-free transfer considering the $J_2$ perturbation demonstrates the proposed method on a problem that existing methods cannot easily handle.

\subsection{Two-Impulse Phase-Free Transfer Maneuvers in Keplerian Orbits}
\label{sec:ex_kep}
We first demonstrate the proposed framework for solving for a minimum-$\Delta V$ two-impulse phase-free maneuver in a Keplerian two-body system. This is a well studied problem, with analytical solutions such as the Hohmann or bi-elliptic transfer in planar cases. Recently, \cite{koblick2019robust} proposed a method to systematically compute solutions for this problem.
Although this problem has been well studied, it is a useful scenario to demonstrate the proposed method and derive the results of \cite{koblick2019robust} from a geometric perspective.

We consider a Keplerian two-body system, consisting of a primary body of mass $M$ and a secondary body of mass $m$, where $ M \gg m $ so that the primary is treated as stationary at the origin. 
The motion of the secondary body is governed by Newtonian gravitation, leading to the classical Kepler problem.

The equations of motion for the secondary body are
\begin{equation}
\ddot{\bs{r}} = -\frac{GM}{r^3} \bs{r},
\end{equation}
where \( \bs{r} \) is the position vector of the secondary body relative to the primary, \( r = \|\bs{r}\| \) is the distance between the two bodies, and $G$ is the universal gravitational constant.

The goal of the minimum $\Delta V$ problem is to minimize the sum of the velocity impulse vectors that transfer a spacecraft from its initial orbit at position $\bs{p}_0$ with an impulsive change in velocity $\Delta \bs{V}_0$ onto the transfer orbit, and, ultimately, to the target orbit at position $\bs{p}_f$ with a second impulsive change in velocity $\Delta \bs{V}_f$, given by
\begin{equation}
\Delta V =\min \Bigl (||\Delta \bs{V}_0|| + ||\Delta \bs{V}_f|| \Bigr).
\end{equation}

In a Keplerian system, given the mechanical energy $E$ of the spacecraft, the Jacobi metric $\mathcal{G}$ is defined by
\begin{equation}
\mathcal{G}(r) = 2\left(E + \frac{GM}{r}\right) I,
\end{equation}
where $I$ is the $3 \times3$ identity matrix. 
The geodesics of this metric correspond exactly to the trajectories of the particle under the gravitational field at a fixed energy $E$. In this case, we know that the Keplerian trajectories are ellipses when $E<0$ based on the derivation of the Jacobi metric (see \cref{sec:appendix_derivation_jacobi}); thus we know that the geodesics are also ellipses. 
Using this, we can circumvent the lengthy process of determining the geodesics through the heat flow PDE with Christoffel symbols and instead solve a simple geometry problem. 
Namely, the geodesic can be computed by finding an ellipse through the given points $\bs{p}_0, \bs{p}_f$ with one focus at the primary body, in the plane defined by $\bs{p}_0, \bs{p}_f$ and the primary body with a free parameter of the semi-major axis $a$. We can relate $a$ to the energy of the orbit $E$ with
\begin{equation}E=\frac{-GM}{2a}.\end{equation}

To bound our sampling approach, we determine a lower bound on the energy $E$ by finding the minimum-energy ellipse that intersects the two points $\bs{p}_0$ and $\bs{p}_f$ \cite{mcclain2001fundamentals}
The semi-major axis $a_{min}$ for the minimum-energy ellipse is 
\begin{equation}
a_{\min}(\bs{p}_0, \bs{p}_f) = \frac{1}{4} \left( \|\bs{p}_0\| + \|\bs{p}_f\| + \|\bs{p}_0 - \bs{p}_f\| \right). 
\end{equation}

Based on this construction, the minimum-$\Delta V$ problem for a specified pair of points $\bs{p}_0, \bs{p}_f$ can be reduced to an optimization over a single parameter: $E$. One could instead explicitly formulate a function of $E$ and compute its minimum, much like the approach taken in \cite{koblick2019robust}, which was parameterized over the angular momentum of the transfer orbit. 
We use this example primarily to illustrate the operation of our general framework: by sampling and refinement alone, we can recover comparable results without requiring an explicit function for optimization. 
This strategy naturally extends to more complex settings where the geodesics are sufficiently complicated such that no simple one-parameter optimization can be formulated, such as the example explored in \cref{sec:J2}.

To improve performance without excessively fine sampling, we implemented a general coarse-to-fine algorithm. With an initial evenly spaced coarse sampling of $\bs{p}_0$, $\bs{p}_f$ across both the initial and target orbits coupled with a sparse sampling in the transfer orbit energy $E$, we ensure sampling coverage in all possible basins that may contain minima, and prevent convergence to a local minimum. We then take a select few transfer trajectories of the smallest $\Delta V$ as ``best candidates'' to initialize an evolutionary algorithm that refines the candidates to obtain the optimal transfer trajectory.

We implemented the procedures described above in Python. We used the PyGMO library from the European Space Agency \cite{biscani2020pagmo} as the evolutionary algorithm, which has also been used for trajectory optimization in the PyKEP project \cite{izzo2012pygmo}. We performed some light tuning on the hyperparameters, used in all examples below, to achieve roughly a 30-second problem solve runtime. Notable parameters are given in \cref{tab:tuneparams}. For 90 position samples per orbit with 3 energy levels sampled each, repeated in both the prograde and retrograde transfer orbit directions, this yields $48,600$ ellipse samples to evaluate in the coarse sampling stage of the algorithm.

\begin{table}[h!]
\centering
\caption{Tuned hyperparameters of coarse-to-fine algorithm using ellipse.}
\begin{tabular}{l r}
\toprule
\toprule
Parameter & Value \\
\hline
\multicolumn{2}{c}{Coarse sample stage}\\
\hline
Number of position samples per orbit & 90 \\
Number of energy level samples per position sample & 3 \\
Number of best to save for refinement & 5 \\
Multiplier on $a_{min}$ for upper bound & 2 \\
\hline
\multicolumn{2}{c}{Evolutionary algorithm refinement stage}\\
\hline
Generations of evolution & 1000 \\
Optimization algorithm & CMA-ES \\
Convergence tolerance & $10^{-12}$ \\
Use generalized monotonic basin hopping (MBH) & True \\
MBH maximum step size & 0.05 \\
\hline
\end{tabular}
\label{tab:tuneparams}
\end{table}

\begin{table}[h!]
\centering
\caption{Tuned hyperparameters of refine-only algorithm using ellipse.}
\begin{tabular}{l r}
\toprule
\toprule
Parameter & Value \\
\hline
\multicolumn{2}{c}{Evolutionary algorithm refinement stage}\\
\hline
Population size to evolve & 50 \\
Generations of evolution & 5000 \\
Optimization algorithm & CMA-ES \\
Convergence tolerance & $10^{-12}$ \\
Use generalized MBH & False \\
\hline
\end{tabular}
\label{tab:tuneparams_gmo}
\end{table}

With this approach, we demonstrate our results on some benchmark examples, with a Low Earth Orbit (LEO) to a Highly Eccentric Orbit (HEO) transfer and Geosynchronous Transfer Orbit (GTO) to Retrograde Geosynchronous Orbit (RGEO) transfer from \cite{koblick2019robust}, and the Earth to asteroid Dionysus problem from \cite{saloglu2023existence}. 

We also compare the solution computed from the elliptical geodesic against the solution computed from the geometric heat-flow PDE to verify the PDE solver used in \cref{sec:J2} for the $J_2$ perturbation case. To manage numerical complexity and computation time of the PDE solver, tolerances and the number of refinement iterations were reduced, with these parameters shown in \cref{tab:tuneparams_J2}. Since we introduce more numerical calculations using a Chebyshev polynomial to approximate the trajectory and to solve a PDE, we expect our final result to perform worse than the direct ellipse calculations for the Keplerian case. However, in other cases without an analytically parameterized solution, this approximation is the best that we can achieve up to the specified convergence tolerance criterion.

Further, we show results implementing a refine-only algorithm that does not initially perform the coarse sampling stage, with notable tuned parameters in \cref{tab:tuneparams_gmo}. In the refine-only algorithm, the initial population is constructed using an evenly spaced position sample space then computing the minimum-energy ellipse between the two sampled points.

We also directly compare the results against two algorithms that can solve the general two-impulse phase-free Keplerian maneuvers. First is an implementation of the algorithm as from Koblick et al. (\verb|semi-analytic|) as presented in \cite{koblick2019robust} and written in Python, which may have a worse runtime performance compared to the original implementation (however, the authors did not explicitly note any run times). Second is the \verb|pl2pl_N_impulses| optimization class from PyKEP (\verb|pykep-pl2pl|)\cite{izzo2012pygmo}, which we have tuned to increase solver robustness. The tuned parameters and implementation notes are elaborated in the Appendix \ref{sec:appendix_other_parameters}. Using the coarse-to-fine algorithm we are able to match or exceed the reported optimal minimum $\Delta V$ presented for each problem from the previous publications.

\subsubsection{Low Earth Orbit (LEO) to a Highly Eccentric Orbit (HEO) Transfer}
\label{example:LEO to HEO}

First, we compare our algorithm with a LEO to HEO orbit transfer of the ALSAT 1 and ARIANE 44 satellites. The initial and target orbits are defined by their position and velocity vectors, which are computed from the orbital elements given in \cite{koblick2019robust}, shown in \cref{tab:ex1_params}).

\begin{table}[h!]
\centering
\caption{Orbit definitions of ALSAT 1 (initial orbit) and ARIANE 44 (target orbit).}
\begin{tabular}{l r r r}
\toprule
\toprule
Parameter & X & Y & Z \\
\hline
Initial position $r_i$ (km) & 3449.16114893 & -2063.72624968 & 5808.89565173 \\
Initial velocity $v_i$ (km/s) & 4.19600114 & -4.65510855 & -4.14528944 \\
\hline
Target position $r_f$ (km) & 7132.67709309 & 644.58087289 & -698.32594990 \\
Target velocity $v_f$ (km/s) & -0.91780300 & 9.52351726 & -0.58384682 \\
\hline
\end{tabular}
\label{tab:ex1_params}
\end{table}

A visualization of the found phase-free minimum-$\Delta V$ transfer solution using the coarse-to-fine algorithm is shown in \cref{fig:ex1_traj}, where the two velocity impulse vectors are shown with black arrows. 
The results of the optimal solution for each algorithm are given in \cref{tab:ex1_results} where the $\Delta V$ of the proposed method is the lowest.

\begin{table}[h!]
\centering
\caption{Results of minimum-$\Delta V$ solution of LEO to HEO transfer.}
\begin{tabular}{l r r r}
\toprule
\toprule
Algorithm & Method type & Total $\Delta V$ (km/s) & Compute time (s) \\
\hline
Coarse-to-fine (ellipse) & Sampling \& optimization & 6.552653 & 24.470 \\
Coarse-to-fine (PDE) & Sampling \& optimization & 6.606830 & 11850.0 \\
Refine-only (ellipse) & Optimization & 6.552653 & 8.381 \\ 
\hline 
\verb|semi-analytic| \cite{koblick2019robust} & Sampling & 6.552675 & 28530.0 
\tablefootnote{Using the same parameters as presented in \cite{koblick2019robust}, including discrete grid search of $0.1^\circ$ and a bisection convergence tolerance of $10^{-12}$. Our implementation of their algorithm in Python is likely much slower than the original, but the authors did not state any explicit run times for comparison.}\\
\verb|pykep-pl2pl| \cite{izzo2012pygmo} & Optimization & 6.612453 & 0.444 \tablefootnote{PyKEP has the majority of its software package implemented in C++, thus is able to achieve significantly faster run times.}\\

\hline
\end{tabular}
\label{tab:ex1_results}
\end{table}

The total change in velocity previously reported for this transfer was $\Delta V =  6.552674$ km/s from \cite{koblick2019robust}, for which this algorithm used a very fine discrete grid search of $0.1^\circ$ spacing along the orbit true anomaly, requiring $12,960,000$ vector pairs of positions to sample over. As seen in \cref{tab:ex1_results}, we were able to find a solution with a $\Delta V$ that is $2.2\times 10^{-5}$ km/s lower while requiring significantly fewer sample evaluations.

We further examine the effectiveness of the coarse-to-fine algorithm from the contour plots of the optimal total-$\Delta V$ for each combination of true anomalies from the initial orbit to the target orbit in \cref{fig:ex1_contour}. For a direct comparison, the total-$\Delta V$ at each point in the contour is computed using the procedure as described in \cite{koblick2019robust}. In the 2D plot on the right of \cref{fig:all_contours}, a grid is overlaid to indicate the locations where the coarse stage of the algorithm sampled (at the intersections of each line). Further, the 5 best samples that were passed into the refinement stage are plotted with white dots. Lastly, the optimal solution is marked with a red-outlined white star on both contour plots. We can observe in \cref{fig:ex1_contour} that the coarse sampling is still dense enough to cover all features of the non-convex contour, and that the best sample candidates were in the correct basin for finding the global minimum.
\begin{figure}[t!]
    \begin{subfigure}{.3\textwidth}
         \centering
         \includegraphics[width=\textwidth, trim={100 35 75 40},clip]{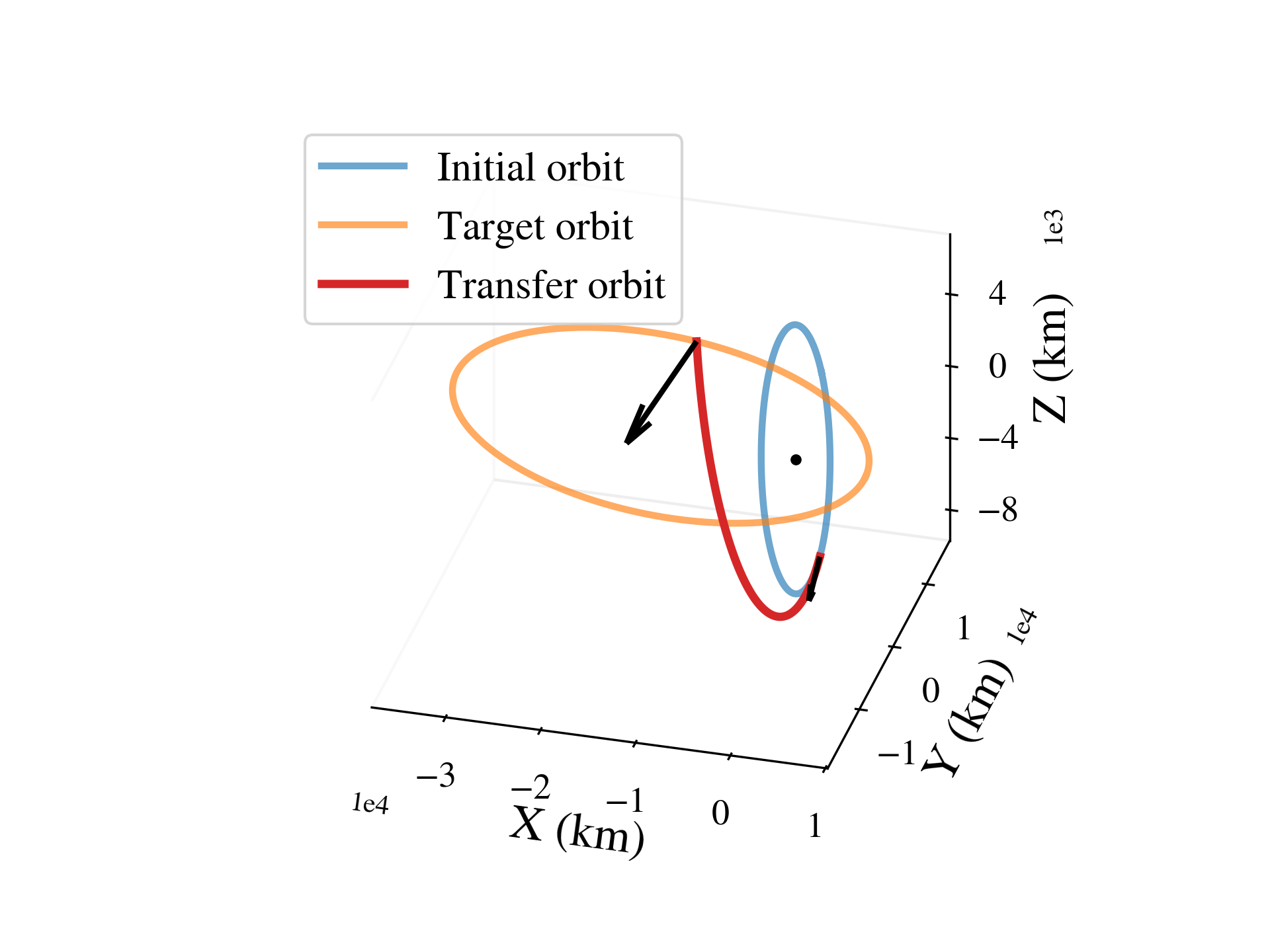}
         \caption{LEO to HEO transfer.}
         \label{fig:ex1_traj}
     \end{subfigure}
     \hfill{}
     \begin{subfigure}{.3\textwidth}
         \centering
         \includegraphics[width=\textwidth, trim={100 35 75 40},clip]{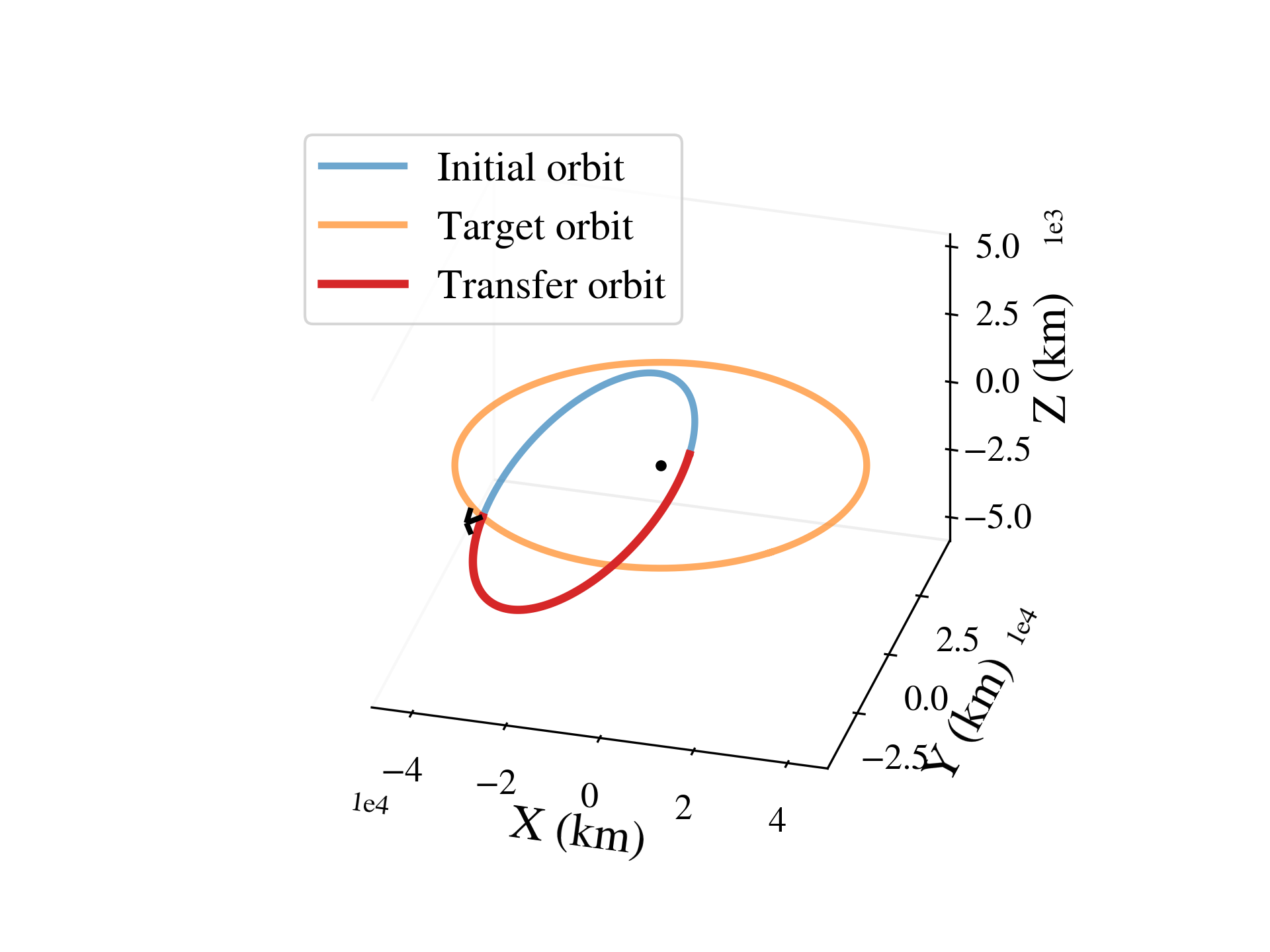}
         \caption{GTO to RGEO transfer.}
         \label{fig:ex2_traj}
     \end{subfigure}
     \hfill{}
     \begin{subfigure}{.3\textwidth}
         \centering
         \includegraphics[width=\textwidth, trim={100 35 75 40},clip]{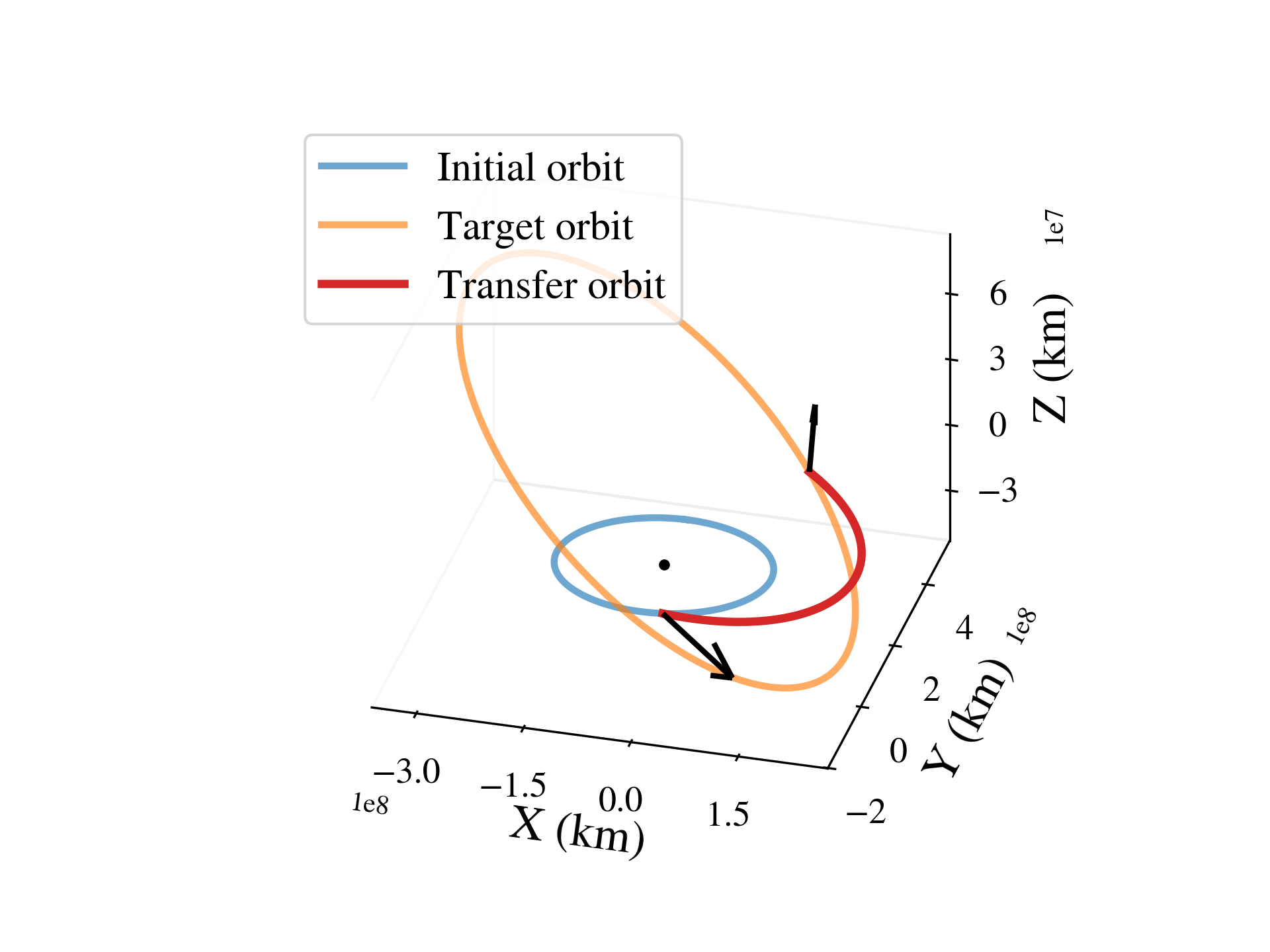}
         \caption{Earth to Dionysus transfer}
         \label{fig:ex3_traj}
     \end{subfigure}
     \caption{Transfer trajectories of Examples 1 - 3.}
     \label{fig:states}
\end{figure}

\begin{figure}
    \begin{subfigure}{.45\textwidth}
         \centering
         \includegraphics[height=.3\textheight, trim={60 0 75 40},clip]{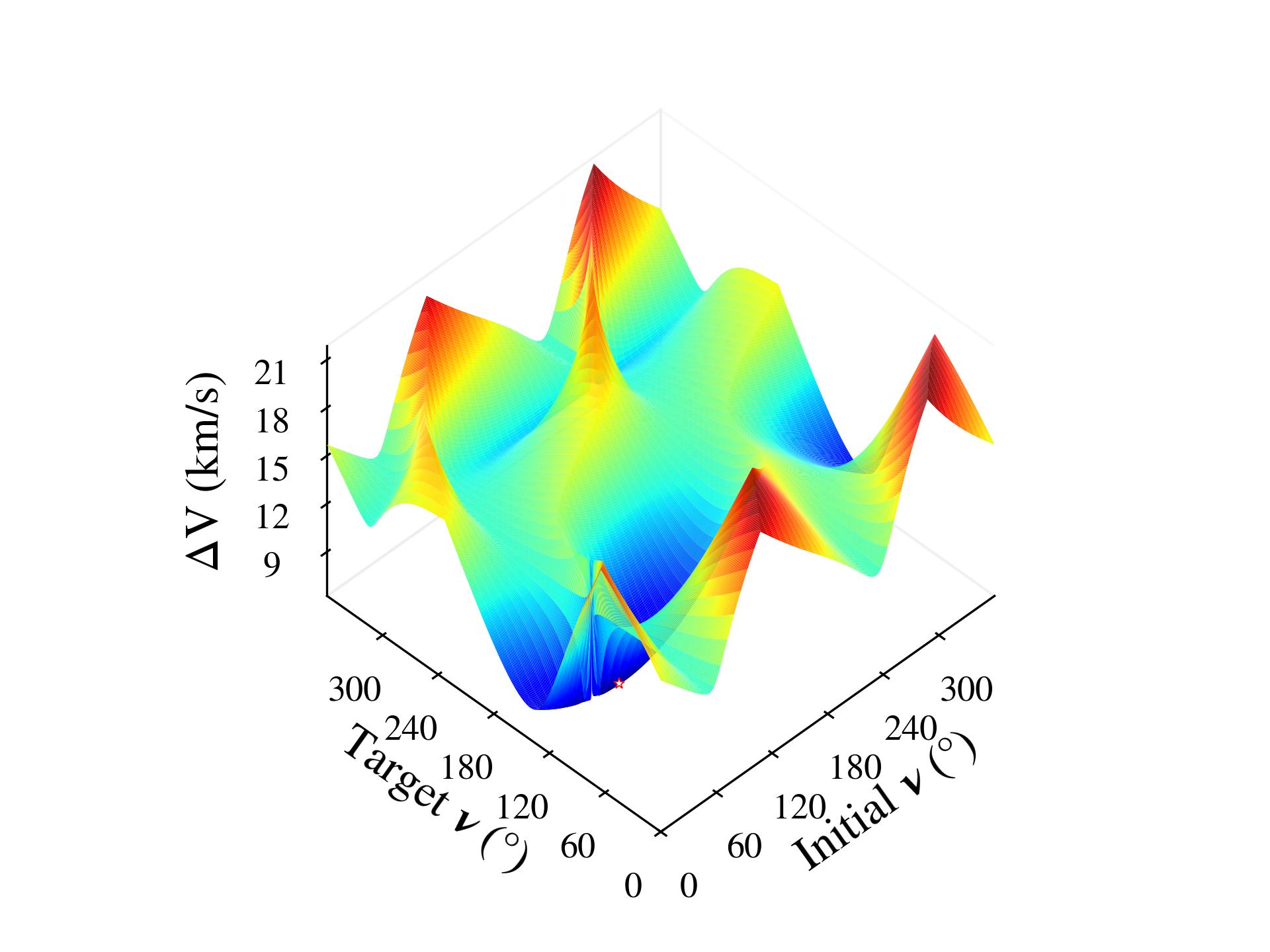}
         \caption{LEO to HEO.}
         \label{fig:ex1_contour3d.png}
     \end{subfigure}
     \hfill{}
     \begin{subfigure}{.45\textwidth}
         \centering
         \includegraphics[height=.3\textheight, trim={57 0 37 20}, clip]{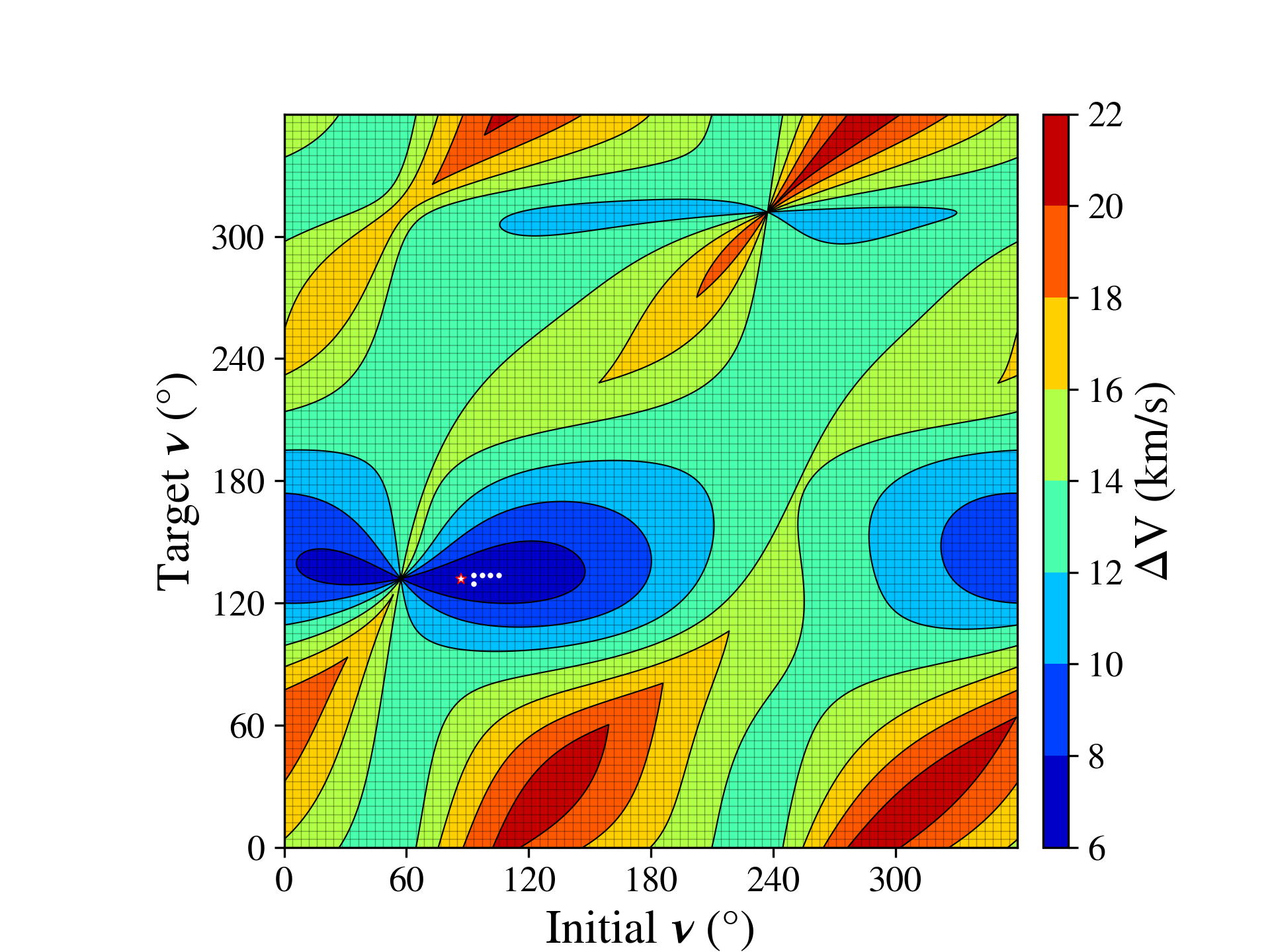}
         \caption{LEO to HEO.}
         \label{fig:ex1_contour}
     \end{subfigure}
     \vfill{}
     \begin{subfigure}{.45\textwidth}
         \centering
         \includegraphics[height=.3\textheight, trim={60 0 75 40},clip]{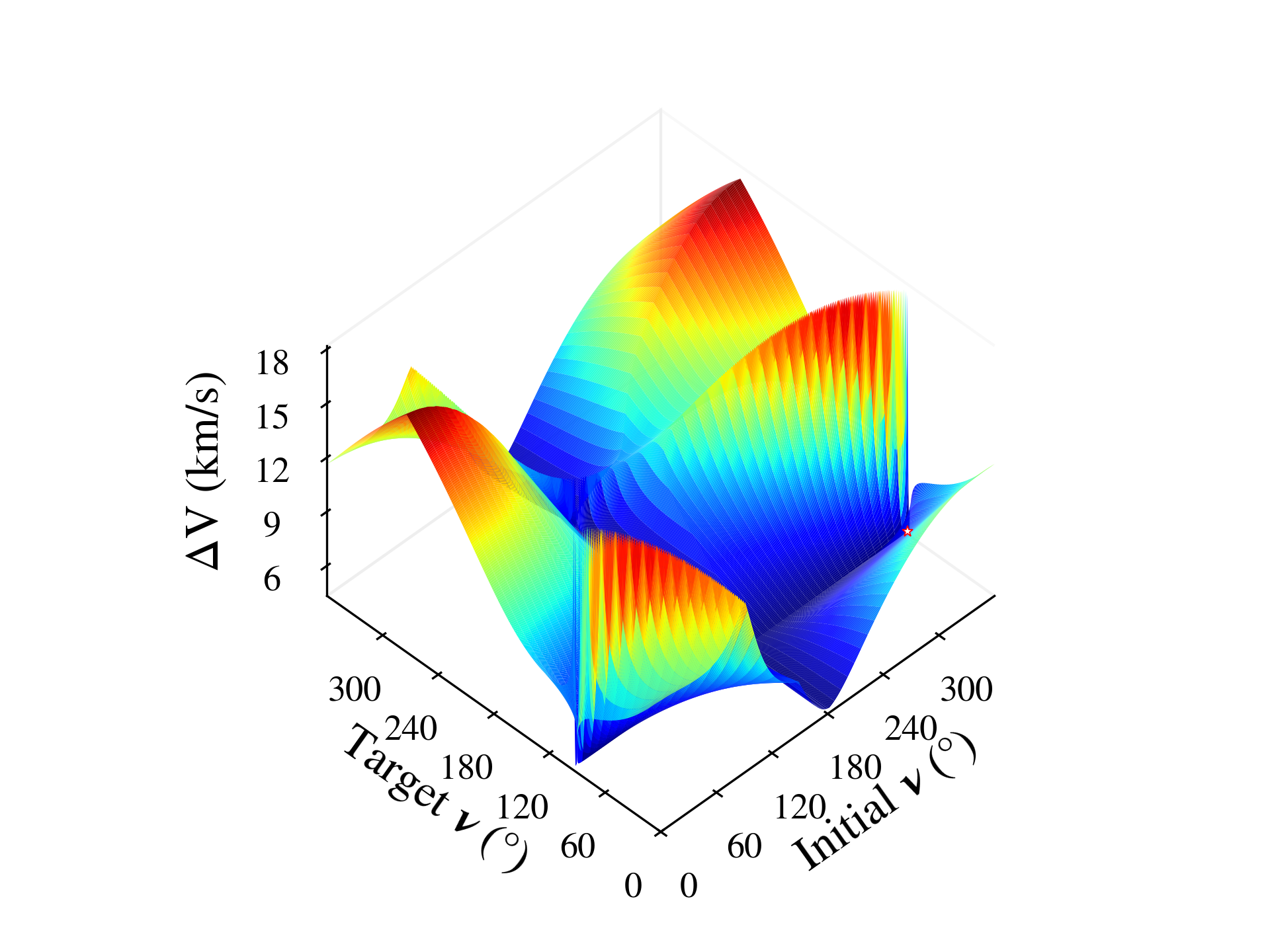}
         \caption{GTO to RGEO.}
         \label{fig:ex2_contour3d}
     \end{subfigure}
     \hfill{}
     \begin{subfigure}{.45\textwidth}
         \centering
         \includegraphics[height=.3\textheight, trim={57 0 37 20},clip]{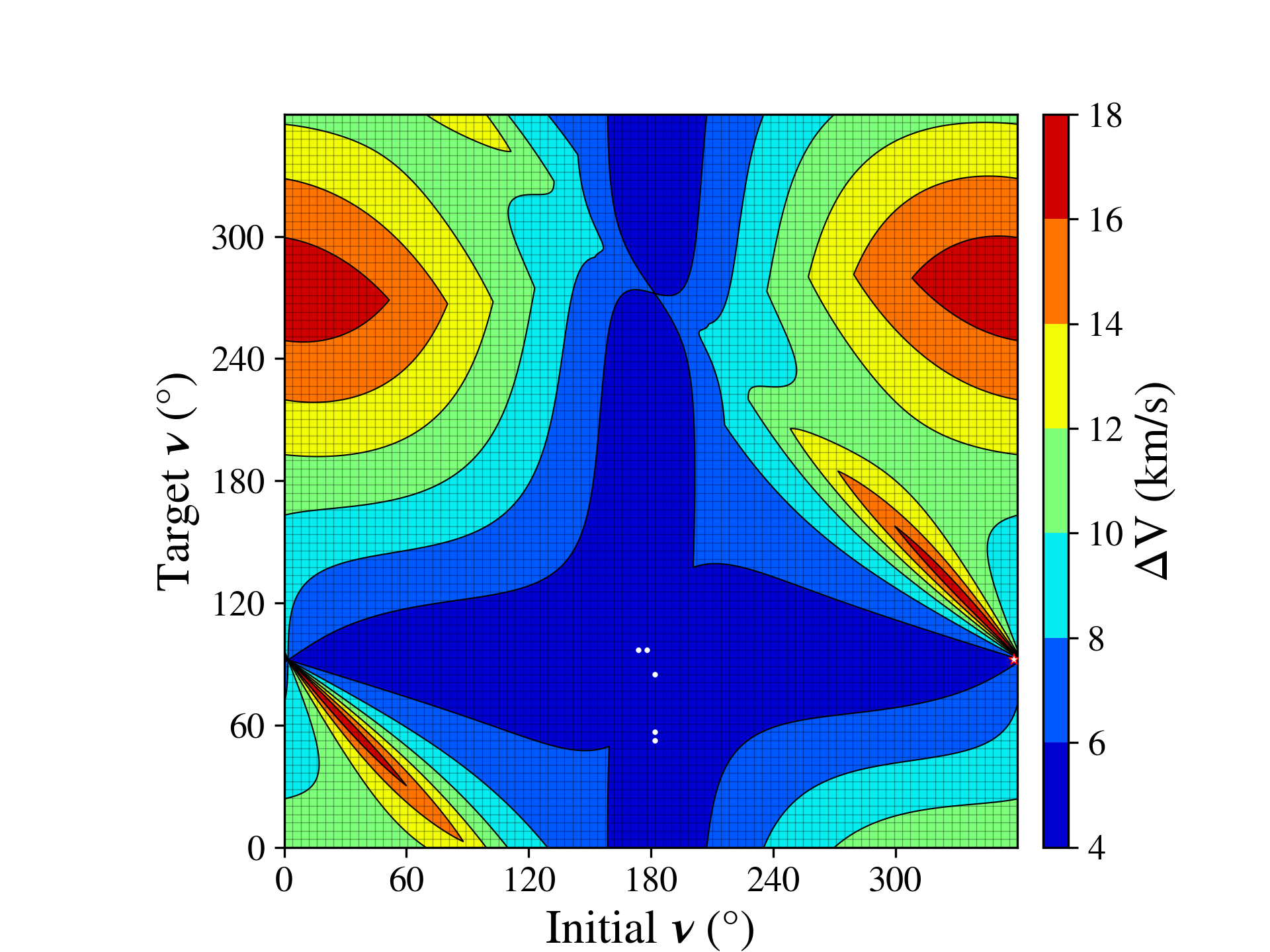}
         \caption{GTO to RGEO.}
         \label{fig:ex2_contour}
     \end{subfigure}
     \vfill{}
     \begin{subfigure}{.45\textwidth}
         \centering
         \includegraphics[height=.3\textheight, trim={60 0 75 40},clip]{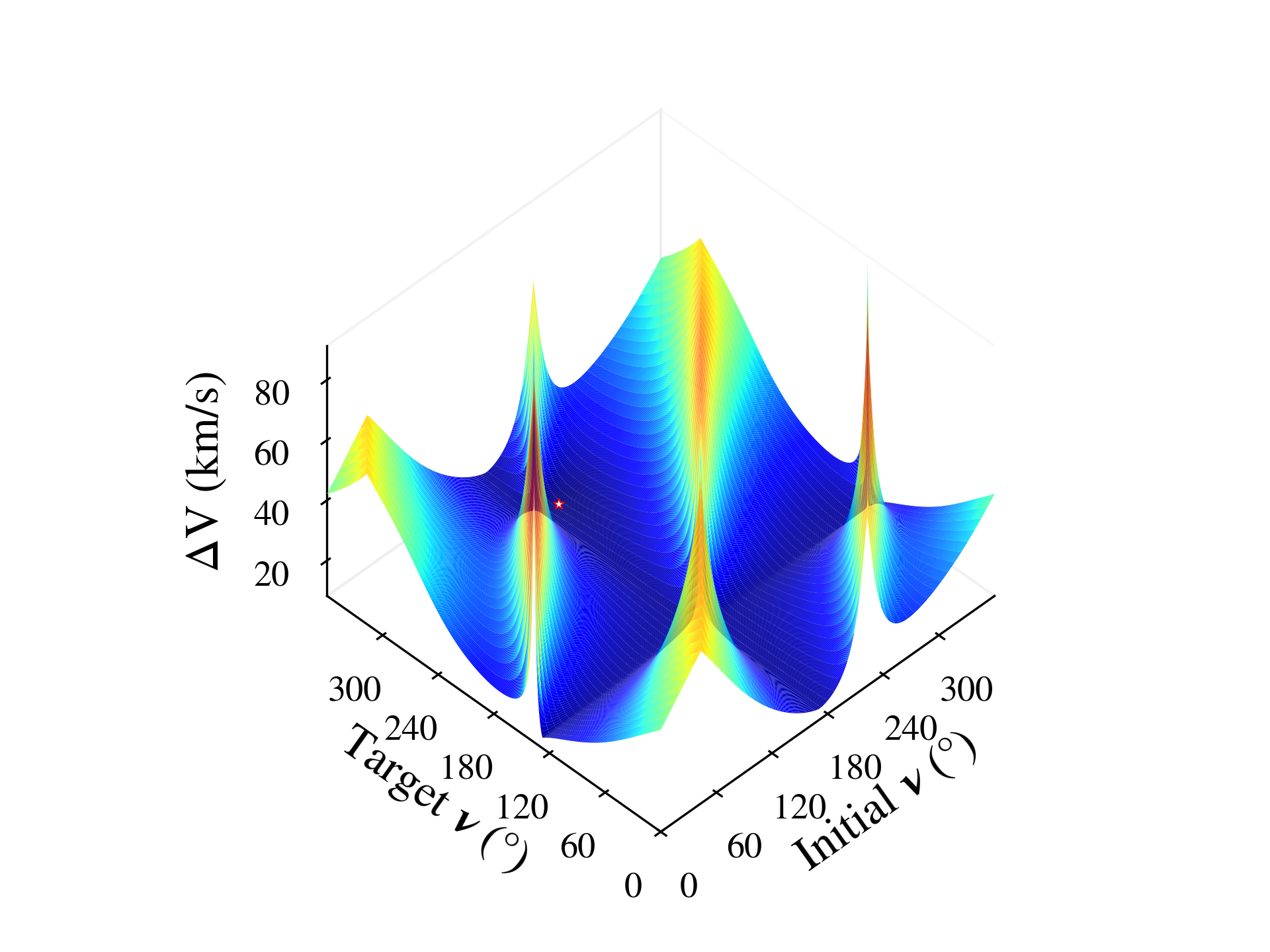}
         \caption{Earth to Dionysus.}
         \label{fig:ex3_contour3d}
     \end{subfigure}
     \hfill{}
     \begin{subfigure}{.45\textwidth}
         \centering
         \includegraphics[height=.3\textheight, trim={57 0 37 20},clip]{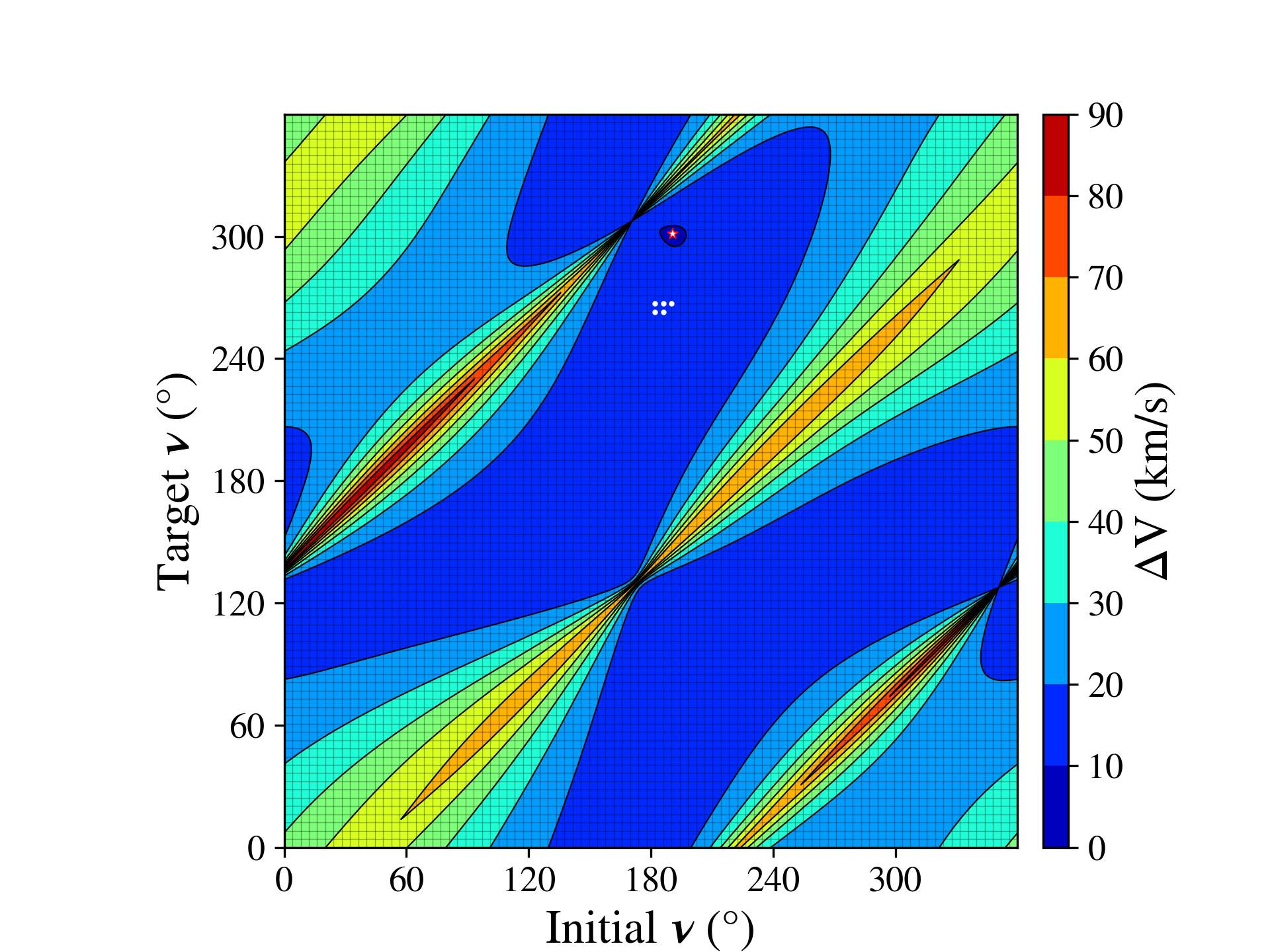}
         \caption{Earth to Dionysus.}
         \label{fig:ex3_contour}
     \end{subfigure}
     \caption{Contours of total $\Delta V$ for initial and target orbit true anomalies of Examples 1 - 3.}
     \label{fig:all_contours}
\end{figure}

\subsubsection{Geosynchronous Transfer Orbit (GTO) to Retrograde Geosynchronous Orbit (RGEO) Transfer}
\label{example:GTO to RGEO}

We next look at a GTO to RGEO transfer, where the orbits are aligned to create a special case where multiple solutions of equally optimal minimum-$\Delta V$ transfers exist \cite{koblick2019robust}, visualized in a contour plots in \cref{fig:ex2_contour}. \cref{tab:ex2_params} gives the position and velocity vectors defining the initial and target orbits from \cite{koblick2019robust}. This is particularly challenging for convergence in the refinement stage of the algorithm. Since we use basin hopping in the evolutionary algorithm, convergence can be slower as there are multiple manifolds of equally optimal solutions.

\begin{table}[h!]
\centering
\caption{Orbit definitions of GTO (initial orbit) and RGEO (target orbit).}
\begin{tabular}{l r r r}
\toprule
\toprule
Parameter & X & Y & Z \\
\hline
Initial position $r_i$ (km) & 4783.85656098 & 4478.04491028 & 74.45791683 \\
Initial velocity $v_i$ (km) & -6.60516350 & 7.11177002 & -3.33974946 \\
\hline
Target position $r_f$ (km) & 30993.40736267 & -28901.81993650 & 0.0 \\
Target velocity $v_f$ (km) & -2.09161279 & -2.24298010 & 0.0 \\
\hline
\end{tabular}
\label{tab:ex2_params}
\end{table}

\begin{table}[h!]
\centering
\caption{Results of minimum-$\Delta V$ solution of GTO to RGEO transfer.}
\begin{tabular}{l r r r}
\toprule
\toprule
Algorithm & Method type & Total $\Delta V$ (km/s) & Compute time (s) \\
\hline
Coarse-to-fine (ellipse) & Sampling \& optimization & 4.600609 & 24.340 \\
Coarse-to-fine (PDE) & Sampling \& optimization & 4.606225 & 11270.0 \\
Refine-only (ellipse) & Optimization & 4.606196 & 4.282 \\
\hline 
\verb|semi-analytic| \cite{koblick2019robust} & Sampling & 4.601186 & 30910.0 \\
\verb|pykep-pl2pl| \cite{izzo2012pygmo} & Optimization & 4.605652 & 0.282 \\
\hline
\end{tabular}
\label{tab:ex2_results}
\end{table}

We observe from \cref{fig:ex2_contour} that the best candidates from the initial coarse sampling stage were found along the vertical manifold near $180^\circ$ true anomaly on the initial orbit, whereas the final converged solution from the refinement converged very far away from these candidates.

 Koblick et al. reported $\Delta V =  4.600604$ km/s for this transfer using the same fine discrete sampling grid \cite{koblick2019robust}. Our solution, which required far fewer computations, has a $\Delta V$ that is $8\times10^{-6}$ km/s higher, which we consider to be sufficiently close given the special case of the existence of two manifolds of equally-optimal $\Delta V$ transfers.

\subsubsection{Earth to Asteroid Dionysus Transfer}
\label{example:Earth to Dionysus}

Lastly, we demonstrate our algorithm on a problem with a largely different scale: a transfer from an Earth orbit to the asteroid Dionysus, which has been analyzed extensively by \cite{saloglu2023existence} in the case of multiple impulse transfers. The position and velocity vectors defining these initial and target orbits are displayed in \cref{tab:ex3_params}.

\begin{table}[h!]
\centering
\caption{Orbit definitions of Earth (initial orbit) and asteroid Dionysus (target orbit).}
\begin{tabular}{l r r r}
\toprule
\toprule
Parameter & X & Y & Z \\
\hline
Initial position $r_i$ (km) & -3637871.081 & 147099798.784 & -2261.441 \\
Initial velocity $v_i$ (km/s) & -30.265 & -0.848 & 5.050 $\times \ 10^{-5}$ \\
\hline
Target position $r_f$ (km) & -302452014.884 & 316097179.632 & 82872290.075 \\
Target velocity $v_f$ (\unit{km/s}) & -4.533 & -13.110 & 0.656 \\
\hline
\end{tabular}

\label{tab:ex3_params}
\end{table}

\begin{table}[h!]
\centering
\caption{Results of minimum-$\Delta V$ solution of Earth to Dionysus transfer.}
\begin{tabular}{l r r r}
\toprule
\toprule
Algorithm & Method type & Total $\Delta V$ (km/s) & Compute time (s) \\
\hline
Coarse-to-fine (ellipse) & Sampling \& optimization & 9.907426 & 42.920\\
Coarse-to-fine (PDE) & Sampling \& optimization & 9.998698 & 10850.0 \\
Refine-only (ellipse) & Optimization & 10.312696 & 24.354 \\
\hline 
\verb|semi-analytic| \cite{koblick2019robust}& Sampling & 9.907427 & 77680.0 \\
\verb|pykep-pl2pl| \cite{izzo2012pygmo} & Optimization & 11.545156 & 0.393 \\

\hline
\end{tabular}
\label{tab:ex3_results}
\end{table}

The optimal total impulse reported previously for this transfer was $\Delta V = 9.907425$ km/s from \cite{saloglu2023existence}, where our proposed method was $1\times10^{-6}$ km/s greater, without changing the parameters from the Earth-centered examples.
\cref{tab:ex3_results} shows that the optimization-based approaches yield a suboptimal result, likely due to converging to a local minimum.

\subsection{Two-Impulse Phase-Free Transfer Maneuvers in Orbits with $J_2$ Perturbation Effects}
\label{sec:J2}
The $J_2$ perturbation results from an approximation of the gravitational effects of a planet that is not perfectly spherical but slightly oblate, with a larger diameter at the equator.
This asymmetry changes the gravitational potential field, resulting in precessing orbits \cite{mcclain2001fundamentals}.
An example orbit modeled with the $J_2$ perturbation can be seen in \cref{fig:J2_all} (blue curve), which shows the orbit precesses because of the perturbation. 
Since the $J_2$ perturbation represents a modification of the potential field, it can be readily incorporated into the geometric framework by defining a new Jacobi metric for this potential field.
The total gravitational potential including the $J_2$ perturbation is given by
\begin{equation}
U(\mathbf{\bs{r}}) = -\frac{GM}{r} \left( 1 - J_2 \left( \frac{R_{\text{body}}}{r} \right)^2 \frac{3z^2 - r^2}{2r^2} \right),
\end{equation}
where $r=\|\bs{r}\|$, $z$ is the vertical component of $\bs{r}$, and $R_\text{body}$ is the radius of the primary body. 
The Jacobi metric with the $J_2$ perturbation then becomes
\begin{equation}
\mathcal{G}(\bs{r})= 2\left( E + \frac{GM}{r} \left( 1 - J_2 \left( \frac{R_{\text{body}}}{r} \right)^2 \frac{3z^2 - r^2}{2r^2} \right) \right) I,
\end{equation}
where $I$ is the $3\times3$ identity matrix.

The geodesics for this metric cannot be computed in closed-form as was done in the non-perturbed case.
Instead, we leveraged the geometric heat flow approach from \cref{eq:pde_christoffel} to compute the geodesics numerically. 
The Christoffel symbols in this case were computed in Mathematica using \cref{eq:christoffel_def} and then imported into Python.
To solve the PDE, a psedeuspectoral approach was implemented in Python with Chebychev polynomials, and SciPy was used to evolve the ODE at each node.
Since there are two ways of transferring around the central body, i.e., two homotopy classes, two geodesics must be computed for each sampling parameter ($E$, $\bs{p}_0$, and $\bs{p}_f$) and during the refinement stage to find the optimal $\Delta V$ transfer. 
Note that this process is more involved than the non-perturbed case because the geodesics are no longer ellipses and must be computed numerically.
Additionally, we found that most of hyperparameters of the coarse-to-fine algorithm used in the non-perturbed case are largely transferable to the $J_2$ case, which shows the method is not hypersensitive to changes in the potential field. 
Of the parameters that had to be changed, the convergence tolerance and generations for evolution were reduced to accommodate the increased computational cost and slight reduction in accuracy of computing geodesics numerically.  
The hyperparameters used are summarized in \cref{tab:tuneparams_J2}.

\begin{table}[h!]
\centering
\caption{Tuned hyperparameters of coarse-to-fine algorithm using PDE.}
\begin{tabular}{l r}
\toprule
\toprule
Parameter & Value \\
\hline
\multicolumn{2}{c}{Coarse sample stage}\\
\hline
Number of orbit periods & 2 \\
Number of position samples per orbit period & 90 \\
Number of energy level samples per position sample & 3 \\
Number of best to save for refinement & 5 \\
Multiplier on $E_{min}$ for upper bound & 2 \\
Number of Chebyshev polynomial nodes & 25 \\
\hline
\multicolumn{2}{c}{Evolutionary algorithm refinement stage}\\
\hline
Number of Chebyshev polynomial nodes & 30 \\
Generations of evolution & 100 \\
Optimization algorithm & CMA-ES \\
Convergence tolerance & $10^{-10}$ \\
Use generalized MBH & True \\
MBH maximum step size & 0.05 \\
\hline
\end{tabular}
\label{tab:tuneparams_J2}
\end{table}

\subsubsection{$J_2$ - Perturbed Jupiter Orbit to Io Altitude Transfer}
\label{example:Jupiter to Io}

The method outlined above was tested on an orbit transfer around Jupiter as its $J_2$ effects are very prominent.
The parameters for the initial and target orbits are presented in \cref{tab:jupiterex_traj}.
The initial orbit was chosen because of the large influence of the $J_2$ perturbation, while the target orbit is the altitude of Jupiter's moon Io.
The initial orbit was propagated for two periods; as the orbit is precessing, one period no longer captures the full geometry. 
More periods can be considered at the expense of a quadratically increasing number of geodesics to compute.
As seen in \cref{fig:J2_all}, the initial orbit is highly affected by the $J_2$ perturbation, so methods that do not account for the perturbation would not be effective.
For example, if one tried to use the optimal two-impulse trajectory not accounting for the $J_2$ perturbation shown as the green line in \cref{fig:exJ2_traj}, the trajectory misses the target orbit by $1000$ km making it infeasible for a two impulse transfer.
The convergence of the method to the optimal geodesic (transfer orbit) is shown in \cref{fig:exJ2_geo}.
To find the other geodesic using numerical methods, the solver was initialized with an initial guess containing a point reflected across the origin.
The reflected initial guess converges to another geodesic as shown in \cref{fig:exJ2_geo_refl}.
An initial and reflected geodesic was found for all points as sometimes the $\Delta V$ from the other geodesic is lower than that of the initial geodesic.
We computed the minimum $\Delta V$ to transfer between the initial and target orbits to be $\Delta V =9.045339$ km/s. The transfer trajectory is shown as the red line in \cref{fig:exJ2_traj}.
For this entire process, we only used the perturbed geodesics and did not need to use the standard Keplerian system for any computations, demonstrating the generality and effectiveness of our approach.

\begin{figure}[t!]
    \begin{subfigure}{.3\textwidth}
         \centering
         \includegraphics[width=\textwidth, trim={100 35 65 40},clip]{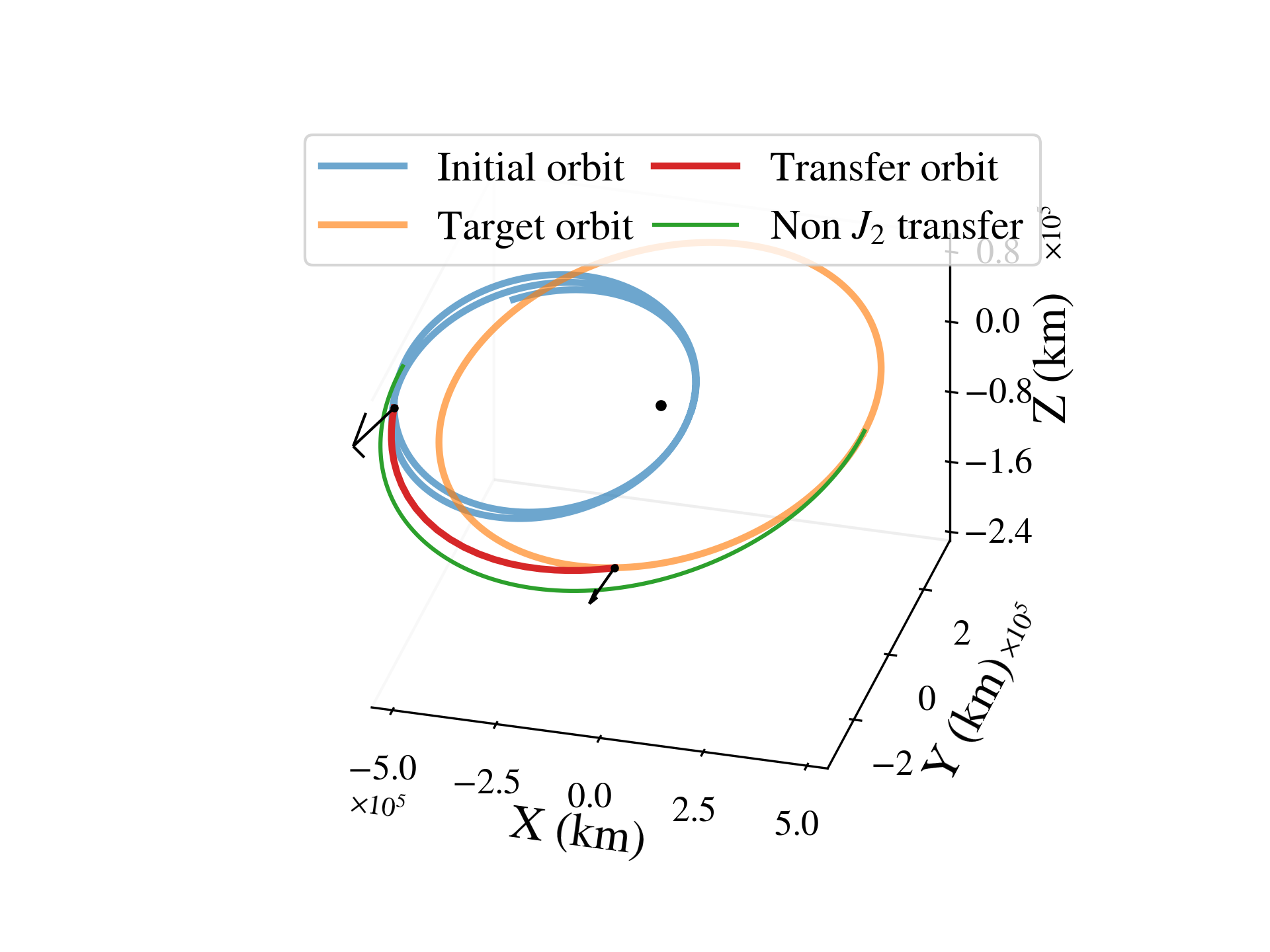}
         \caption{J2-Perturbed Jupiter Orbit to Io altitude transfer.}
         \label{fig:exJ2_traj}
     \end{subfigure}
     \hfill{}
     \begin{subfigure}{.3\textwidth}
         \centering
         \includegraphics[width=\textwidth, trim={100 35 65 40},clip]{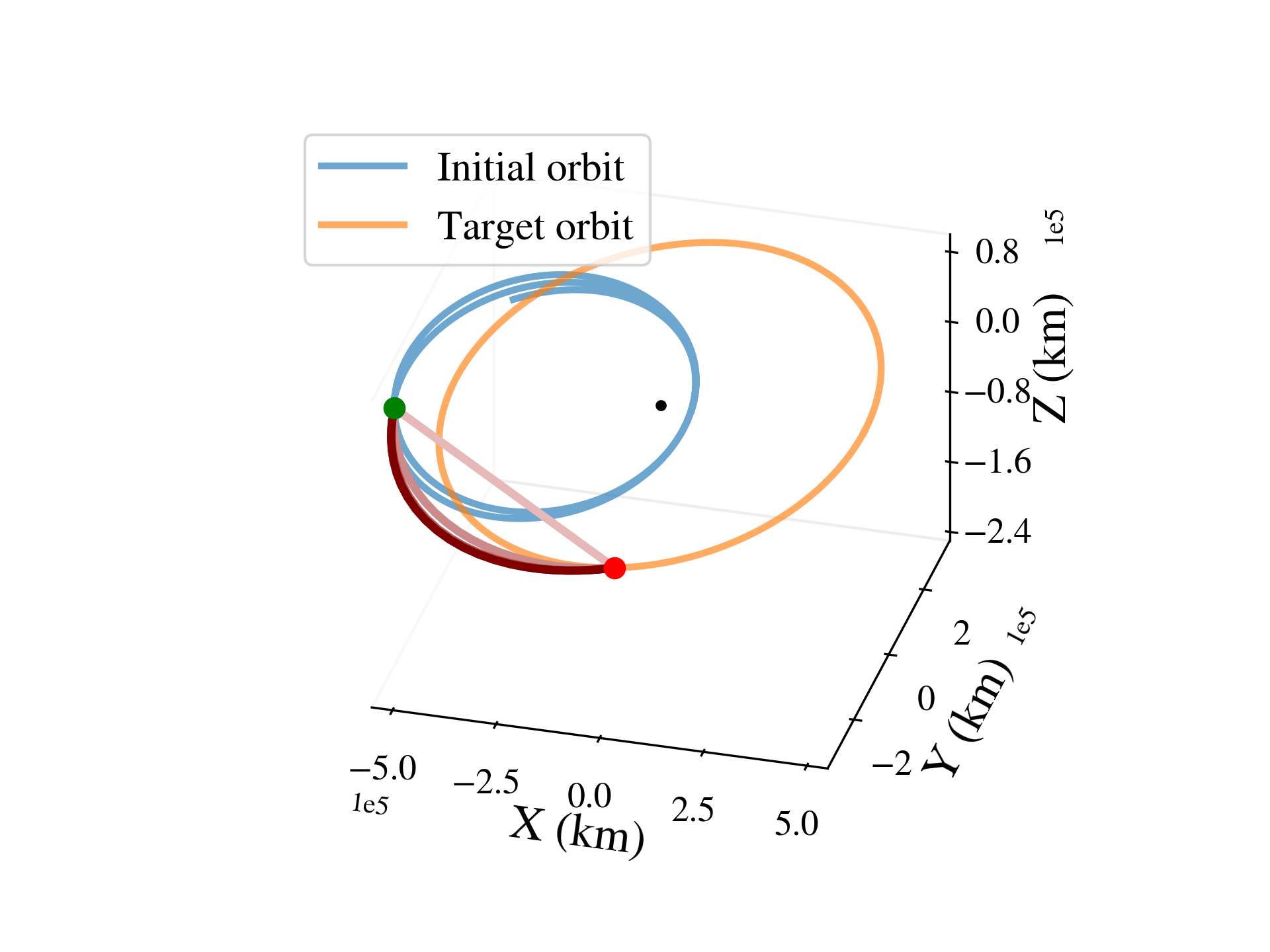}
         \caption{Geodesic solver converging to one geodesic for the same transfer.}
         \label{fig:exJ2_geo}
     \end{subfigure}
     \hfill{}
     \begin{subfigure}{.3\textwidth}
         \centering
         \includegraphics[width=\textwidth, trim={100 35 65 40},clip]{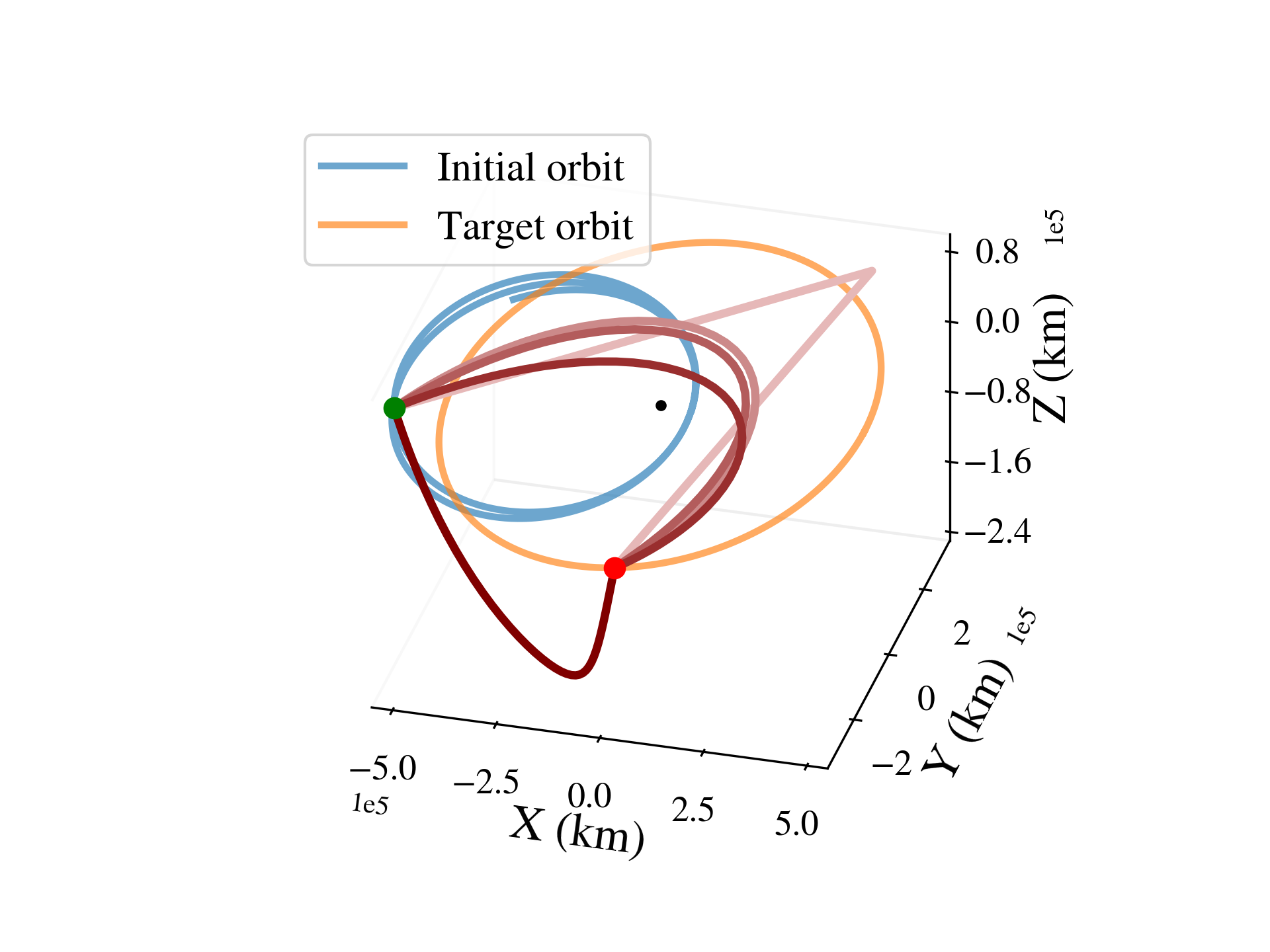}
         \caption{Geodesic solver converging to another geodesic for the same transfer.}
         \label{fig:exJ2_geo_refl}
     \end{subfigure}
     \caption{Transfer trajectory showing the geodesic solving progression}
     \label{fig:J2_all}
\end{figure}

\begin{table}[h!]
\centering
\caption{Definitions of orbit about Jupiter (initial orbit) and orbit about Jupiter at Io's altitude (target orbit).}
\begin{tabular}{l r r r}
\toprule
\toprule
Parameter & X & Y & Z \\
\hline
Initial position $r_i$ (km) & 75000.0 & 0.0 & 0.0 \\
Initial velocity $v_i$ (km/s) & 0.0 & 53.261749 & 14.271442 \\
\hline
Target position $r_f$ (km) & 489943.356 & 0.0 & 0.0 \\
Target velocity $v_f$ (km/s) & 0.0 & 16.112383 & 1.406071$ \times10^{-2}$ \\
\hline
\end{tabular}
\label{tab:jupiterex_traj}
\end{table}

\section{Conclusion}
We presented a new framework for designing two-impulse trajectories in restricted two-body systems using a geometric framework.
By taking a geometric approach, we were able to avoid issues of sensitivity in the initial guess inherent to other approaches, while being able to handle a broader class transfer problems, e.g., $J_2$ perturbations. 
We demonstrated the framework both on the well-studied bi-impulsive Keplerian system with examples from prior studies, where the new framework matches or exceeds the performance of other methods.
We then demonstrated the framework on a two-impulse phase-free transfer between $J_2$-perturbed orbits around Jupiter to show how the framework can be applied to more complex problems without relying on a prior solution of the simpler problem.
Moreover, minimal adjustments were needed for applying the method to the $J_2$ case, showcasing the framework's ability to handle more complicated systems.
Future work includes extending the geometric framework to more complex scenarios, such as the restricted three-body problem.
Since a Jacobi metric cannot be defined in this setting, such an extension would require identifying an appropriate alternative metric and developing a corresponding geodesic solver.
Investigating sampling regimes for multi-impulse maneuvers in multi-body systems is also of interest.

\section{Appendix}
\subsection{Derivation of Jacobi metric for system with natural Lagrangian}
\label{sec:appendix_derivation_jacobi}
Here, we show the derivation of the Jacobi metric for a system with a natural Lagrangian.

\begin{theorem}
  Let the Lagrangian of a mechanical system with configuration coordinates $\bs{q}$ be defined as
  \[
    L\bigl(\bs{q},\dot{\bs{q}} \bigr)
    = \frac12\,\dot{\bs{q}}^{\top}\,\mathcal{Q}(\bs{q})\,\dot{\bs{q}}
      -V(\bs{q}),
  \]
  where $\mathcal{Q}$ is the mass tensor and $V$ the potential. 
  For fixed energy, $E$, the trajectories of this Lagrangian system coincide with the geodesics of the Jacobi metric, $\mathcal{G}(\bs{q}) = 2\,(E - V(\bs{q}))\,\mathcal{Q}(\bs{q})$.
\end{theorem}

\begin{proof}
The standard action is
\[
S
= \int_{t_0}^{t_f} L\bigl(\bs{q}(t),\dot{\bs{q}}(t)\bigr)\,dt
= \int_{t_0}^{t_f}
  \Bigl(\tfrac12\,\dot{\bs{q}}^{\top}\,\mathcal{Q}(\bs{q})\,\dot{\bs{q}} - V(\bs{q})\Bigr)\,dt.
\]
By the Legendre transform,
\[
\bs{\mathfrak{p}}
= \frac{\partial L}{\partial \dot{\bs{q}}}
= \mathcal{Q}(\bs{q})\,\dot{\bs{q}},
\]
the Hamiltonian becomes
\[
H\bigl(\bs{q},\bs{\mathfrak{p}}\bigr)
= \bs{\mathfrak{p}}^{\top}\dot{\bs{q}} - L(\bs q,\dot{\bs q})
= \tfrac12\,\dot{\bs q}^{\top}\,\mathcal{Q}(\bs{q})\,\dot{\bs q} + V(\bs q),
\]
which is constant for a fixed energy level $E$.  Hence
\[
L(\bs{\mathfrak{p}},\dot{\bs q}) = \bs{\mathfrak{p}}^{\top}\dot{\bs q} - E
\quad\Longrightarrow\quad
S = \int_{t_0}^{t_f}\bigl(\bs{\mathfrak{p}}^{\top}\dot{\bs q} - E\bigr)\,dt.
\]
Dropping the constant term $-E(t_f - t_0)$, we get the abbreviated action
\[
S_{\mathrm{abbr}}
= \int_{t_0}^{t_f}\bs{\mathfrak{p}}^{\top}\dot{\bs q}\,dt
= \int_{p_0}^{p_f}\bs{\mathfrak{p}}^{\top}\,d\bs q,
\]
since $d\bs q = \dot{\bs q}\,dt$.
By energy conservation,
\[
\dot{\bs q}^{\top}\,\mathcal{Q}(\bs q)\,\dot{\bs q} = 2(E - V(\bs q)),
\]
so
\[
\bs{\mathfrak{p}}^{\top}\,d \bs q
= \dot{\bs q}^{\top}\,\mathcal{Q}(\bs q)\,\dot{\bs q}\,dt
= \sqrt{\dot{\bs q}^{\top}\,\mathcal{Q}(\bs q)\,\dot{\bs q}}
  \sqrt{2(E - V(\bs q))}\,dt.
\]
Now re-parameterize the curve by $s\in[0,1]$ so that
\[
\bs q(0)=\bs p_0,\quad \bs q(1)=\bs p_f,\quad
\bs q_s=\frac{d\bs q}{ds}(s),\quad
d\bs q=\bs q_s(s)\,ds.
\]
Thus
\[
S_{\mathrm{abbr}}
= \int_{\bs p_0}^{\bs p_f}\bs{\mathfrak{p}}^{\top}\,d\bs q
= \int_{0}^{1}\bs q_s(s)^{\top}\,\mathcal{Q}(\bs q(s))\, \bs q_s(s)ds
= \int_{0}^{1}
  \sqrt{\bs q_s(s)^{\top}\,\mathcal{Q}(\bs q)\,\bs q_s(s)}
  \sqrt{2(E - V(\bs q(s)))}ds.
\]
Finally, looking at the Jacobi metric
\[
\mathcal{G}(\bs q) = 2\,(E - V(\bs q))\,\mathcal{Q}(\bs q),
\]
Then
\[
S_{\mathrm{abbr}}
= \int_{0}^{1} \sqrt{\bs q_s(s)^{\top}\,\mathcal{G}(\bs q(s))\, \bs q_s(s)}ds,
\]
which is exactly the length functional for metric $\mathcal{G}$. 
Because geodesics minimize this functional, they are fixed-energy trajectories of the original system.
\end{proof}

\subsection{Planetary Parameters}
\label{sec:appendix_planetary_parameters}
Parameters for Earth and Jupiter, used in \cref{example:LEO to HEO}, \cref{example:GTO to RGEO}, \cref{example:Earth to Dionysus}, and \cref{example:Jupiter to Io} are given below, referenced from \cite{williams2024jupiter, williams2024sun, williams2024earth}.
\begin{table}[ht]
\centering
\cprotect\caption{Planetary Parameters}
\begin{tabular}{lll}
\toprule
\toprule
& Parameter & Value \\
\hline
\multirow{2}{*}{\rotatebox[origin=c]{90}
{Sun}} & GM (km$^{3}$/s$^{2}$)& \num{1.327e+11} \\ & $R_{body}$ (km) & \num{6.95e5} \\
\hline
\multirow{2}{*}{\rotatebox[origin=c]{90}
{Earth}} & GM (km$^{3}$/s$^{2}$)& \num{3.986e+5} \\ & $R_{body}$ (km) & 6378.0 \\
\hline
\multirow{3}{*}{\rotatebox[origin=c]{90}
{Jupiter}} & GM (km$^{3}$/s$^{2}$) & \num{126.687e+6} \\ & $J_2$ & \num{1.475e-2} \\ & $R_{body}$ (km)& 69911.0 \\
\end{tabular}
\label{tab:space_parameters}
\end{table}

\subsection{Tuned parameters of other algorithms}
\label{sec:appendix_other_parameters}

The algorithm from Koblick et al. (\verb|semi-analytic|) was implemented in Python and initialized with the same parameters as presented in \cite{koblick2019robust}. The fine discrete grid search of 0.1° spacing along true anomaly $\nu$ requires 12,960,000 vector pairs of positions
to sample over, which incurs a significant runtime penalty as implemented in Python. 

\begin{table}[!h]
\centering
\cprotect\caption{Tuned hyperparameters of \verb|semi-analytic|}
\begin{tabular}{l r}
\toprule
\toprule
Parameter & Value \\
\hline
Number of position samples per orbit & 3600 \\
Convergence tolerance & $10^{-12}$ \\
Singularity tolerance & $10^{-3}$ \\
\hline
\end{tabular}
\label{tab:tuneparams_semianalytic}
\end{table}

The algorithm from Izzo et al. (\verb|pykep-pl2pl|) is taken from the PyKEP software package \cite{izzo2012pygmo}. Modifications to \verb|pl2pl_N_impulses| to enable more robust performance include fixing a bug with the logic for phase-free maneuvers and adding another element to the optimization decision vector that determines whether the transfer trajectory is prograde or retrograde, as we have found that the previous calculation did not always select the most optimal direction. The hyperparameters that we have tuned are consistent with the refine-only algorithm we tested \ref{tab:tuneparams_gmo}.

\begin{table}[h!]
\centering
\cprotect\caption{Tuned hyperparameters of \verb|pykep-pl2pl|}
\begin{tabular}{l r}
\toprule
\toprule
Parameter & Value \\
\hline
Population size to evolve & 50 \\
Generations of evolution & 5000 \\
Optimization algorithm & CMA-ES \\
Convergence tolerance & $10^{-12}$ \\
Use generalized MBH & False \\
\hline
\end{tabular}
\label{tab:tuneparams_pykep}
\end{table}

\FloatBarrier
\bibliography{sample}

\end{document}